\documentclass[10pt,draftcls,onecolumn]{IEEEtran}

\usepackage{mathtools}

\usepackage{amsmath,amssymb,amsfonts}
\usepackage{amsthm}
\usepackage{graphicx,color}
\usepackage{enumitem}
\usepackage{cite}
\usepackage{caption}
\usepackage{subcaption}

\usepackage{wasysym}

\usepackage{mathtools}

\usepackage{tcolorbox}

\usepackage{savesym}
\savesymbol{AND}

\usepackage{algorithmic}
\usepackage{algorithm}

\newtheorem{theorem}{Theorem}[section]
\newtheorem{proposition}[theorem]{Proposition}
\newtheorem{lemma}[theorem]{Lemma}

\newtheorem{remark}[theorem]{Remark}

\newcommand{\longthmtitle}[1]{\mbox{}\textit{(#1).}}

\newcommand{\ldef}{:=}
\newcommand{\rdef}{=:}
\newcommand{\Mc}[1]{\mathcal{#1}}
\newcommand{\MbrR}{\Mc{R}(\Rin)}

\newcommand{\redR}{\mathbf{R}^-}
\newcommand{\integers}{{\mathbb{Z}}}
\newcommand{\natz}{{\mathbb{N}}_0}
\newcommand{\nat}{{\mathbb{N}}}
\newcommand{\rev}{\phi_{\hat{G}G}(r)}

\newcommand{\intersect}{\ensuremath{\operatorname{\cap}}}
\newcommand{\intrangecc}[2]{[#1, #2]_\integers}
\newcommand{\intrangeoc}[2]{(#1, #2]_\integers}
\newcommand{\thh}{\text{th}}
\newcommand{\lstar}{{\lambda_r^i}^*}
\newcommand{\map}[3]{#1:#2 \rightarrow #3}
\newcommand\numberthis{\addtocounter{equation}{1}\tag{\theequation}}
\DeclareMathOperator*{\argmin}{argmin}
\DeclareMathOperator*{\argmax}{argmax}

\newcommand{\fril}{f_r^i(l)}

\newcommand{\tfril}{\tilde{f}_r^i(l)}

\newcommand{\tfr}{\tilde{f}_r}
\newcommand{\tFl}{\tilde{F}_l}
\newcommand{\tdl}{\tilde{d}_l}
\newcommand{\tSl}{\tilde{S}_l}
\newcommand{\pril}{p_r^i(l)}

\newcommand{\bril}{\beta_r^i(l)}

\newcommand{\Rhat}{\mathbf{\hat{R}}}

\newcommand{\Rin}{\mathbf{R}_{I}}
\newcommand{\Rout}{\mathbf{R}_{O}}
\newcommand{\Rtl}{\mathbf{\hat{R}}^{-}}
\newcommand{\hfril}{\hat{f}_r^i(l)}

\newcommand{\hfr}{\hat{f}_r}
\newcommand{\hdl}{\hat{d}_l}
\newcommand{\hbril}{\hat{\beta}_r^i(l)}
\newcommand{\htril}{\hat{t}_r^i(l)}
\newcommand{\hFl}{\hat{F}_l}

\newcommand{\stup}{(r,i,l)}
\newcommand{\tril}{t_r^i(l)}
\newcommand{\Mbf}[1]{\mathbf{#1}}   

\newcommand{\srs}[1]{\sigma(#1)}

\newcommand{\oprocendsymbol}{\hbox{$\bullet$}}
\newcommand{\oprocend}{\relax\ifmmode\else\unskip\hfill\fi\oprocendsymbol}

  {\list{}{\leftmargin=#1\rightmargin=#1}\item[]}%
  {\endlist}

\usepackage{epstopdf}
\AppendGraphicsExtensions{.tif}
\allowdisplaybreaks
\graphicspath{{figures/}}
\usepackage[colorlinks=true, linkcolor=blue, citecolor=blue]{hyperref}
\date{}
\begin{document}

\title{One-Shot Coordination of First and Last Mode Transportation}
\author{Subhajit Goswami \qquad Pavankumar Tallapragada %
  \thanks{This work was partially supported by Robert Bosch Centre for
    Cyber-Physical Systems, Indian Institute of Science. A part of
    this work was presented in the 2019 European Control Conference
    as~\cite{SG-PT:2019-ecc}. }%
  \thanks{S. Goswami and P. Tallapragada are with the Department of
    Electrical Engineering, Indian Institute of
    Science. P. Tallapragada is also with the Robert Bosch Centre for
    Cyber-Physical Systems, Indian Institute of Science. \tt\small
    \{subhajitg, pavant\}@iisc.ac.in }%
}

\maketitle
\begin{abstract}
  In this paper, we consider coordinated control of feeder vehicles
  for first and last mode transportation. The model is macroscopic
  with volumes of demands and supplies along with flows of
  vehicles. We propose a one-shot problem for transportation of demand
  to or from a hub within a fixed time window, assuming the knowledge
  of the demand and supply configurations. We present a unified
  optimization framework that is applicable for both operator profit
  maximization and social welfare maximization. The latter goal is
  useful for applications such as disaster response. The decision
  variables in the optimization problem are routing and allocations of
  the vehicles for different services.

  With K.K.T. analysis we propose an offline method for reducing the
  problem size. Further, we also analyze the problem of maximizing
  profits by optimally locating the supply for a given total supply
  and present a closed form expression of the maximum profits that can
  be earned over all supply configurations for a given demand
  configuration. We also show an equivalence between optimal supply
  location in the first mode problem and the last mode problem. We
  present a model for pricing based on the cost and travel time of the
  best alternate transportation and present necessary conditions for
  the feeder service to be viable. We illustrate the results through
  simulations and also compare the proposed model with a traditional
  vehicle routing problem. Through simulations, we also compare with
  the microscopic version of the problem with the decision variables
  being integers. We demonstrate that the route reduction algorithm
  proposed for the macroscopic formulation is still useful for
  computing nearly optimal solutions to the microscopic problem with
  much improved computational efficiency.
\end{abstract}


\section{Introduction} \label{sec:Intro}

This paper explores the idea of a coordinated and demand responsive
feeder service with a single origin (destination) and multi
destination (origin) setup, over a hard single time-window.  This
problem occurs in scenarios such as evacuation and distribution of
relief materials in response to disasters~\cite{VC-RB-AB:2012} and
transportation for one-time events with a large
foot-fall. Furthermore, this scenario has potential applications in
multi-modal transportation systems where the main mode follows a
schedule and needs to be coordinated with first and last modes to
achieve a coordinated system. This potential application is seen in
many cases like peak-hour single destination
para-transit~\cite{RL-RM:2000}, freight
transportation~\cite{MSS-etal:2014}, express courier
systems~\cite{CB-etal:2002} etc.

\subsubsection*{Literature Review} 


In the context of routing of transportation services, the vehicle
routing problem (V.R.P.) \cite{HM:1989, PT-DV:2002-book,
  FF:2013-book,MWU:2017-book, TV-GL-PM-2019-VRP-Review} assumes a
depot from where one or more vehicles are routed via locations where
one or more types of entities are picked up at their origin(s) and
dropped off at their destination(s). There are also many variations to
the basic V.R.P. like capacitated \cite{FF:2013-book}, multiple origin
and destination \cite{PT-DV:2002-book}, multiple time window
\cite{MWU:2017-book, PT-DV:2002-book} or simultaneous pick-up and
delivery \cite{HM:1989} etc. Sharing some similarities with V.R.P. are
the ride sharing problem~\cite{NA-AE-MS-XW:2012, JA-etal:2017,
  FYV-etal:2018, CCT-CYC:2007} and dial-a-ride
problem~\cite{JFC-GL:2007, SCH-etal:2018, BL-DK-TVW-HAR:2016}. In
these problems, the aim is to match vehicles with the passengers while
maximizing operator's profits in some time windows. Recently, there
has also been a growing interest in routing problems with awareness of
the demand and the fleet. For example, ~\cite{KTS-NHD-DHL:2010,
  FM-etal:2016} explore the problem of optimally dispatching taxis
based on the location of the taxis and customer requests.

While many routing problems deal with discrete vehicles and discrete
entities to be transported, macroscopic models that deal with flows of
vehicles and volumes of demand and supply are also common. Though less
realistic than discrete models, they allow for computationally easier
solutions and greater scope for analysis and higher order planning. In
the context of demand anticipative mobility,
\cite{FR:2017,MS-etal:2018, MP-SLS-EF-DR:2012,RZ-MP:2016,
  AV-WG-SS:2017} aim to match demand and supply by routing autonomous
vehicle flows and maximize throughput in the network through a
steady-state design of the load-balancing and routing flow rates. Much
of the literature on fleet routing is in the context of uni-modal or
single hop transportation services. Though multi-modal transportation
has been extensively studied, dynamic, demand and supply aware first
and last mode service is not sufficiently studied~\cite{SS-NC:2016}.

There have been multiple studies on humanitarian aid
planning. Evacuation and distribution of relief material in disaster
management scenarios form the prime focus of such
studies~\cite{AMC-etal-2012-Review, GGP-RB-2013-Review, MSH-etal-2016,
  MS-RB-QH-2020-Review }. Most of these studies focus on determining
the location of the ware-houses and safe-houses given the
forecasts.~\cite{MA-etal-2015} proposes a multi-depot location-routing
problem (MDLRP). The MDLRP considers network failure with multiple use
of vehicles, and a standard relief time. \cite{MH-etal-2012} analyzes
a last-mile delivery problem to gain insights into the underlying
route structures and provide observations on routing aspects that are
necessary for planning out disaster relief operations. The results are
in the context of three scenarios: time minimization, demand
maximization and service comparability. \cite{JMF-etal-2018} proposes
a multi-criterion last-mile routing problem, while
\cite{AMC-etal-2008} considers a capacitated routing problem which
minimizes the \emph{maximum response time} for relief to reach the
affected regions.  In~\cite{RL-ZJ-2019}, a time-windowed V.R.P. with a
single depot is proposed for the purpose of humanitarian aid
distribution.

\subsubsection*{Contributions}

In this paper, we consider the first-mode or \emph{feed-in} problem
and the last-mode or \emph{feed-out} problem, wherein all the demand
has a common destination and origin, respectively. The following are
the main contributions of this paper.
\begin{itemize}
\item \emph{Model:} We present a macroscopic network flow model for
  the one-shot feed-in and feed-out problems. We call it ``one-shot''
  since there is a single hard time window within which the
  transportation has to occur. Given the demand and supply
  configuration on a network, the optimization problem determines
  optimal routing of the vehicles and servicing of the demand. We also
  consider the combined problem of optimal supply location and
  routing. The proposed model includes both maximization of monetary
  profit and maximization of serviced demand or social welfare as
  special cases, depending on the choice of the parameters.
  
\item \emph{Analysis and algorithm:} We carry out a rigorous
  mathematical analysis of the optimization problem, based on which we
  propose an offline algorithm to reduce the problem size by pruning
  ``non-profitable'' routes and decision variables. We also analyze
  the feed-in and feed-out settings, following which we establish an
  equivalence between the optimal supply location problems in the two
  settings. We also present an analysis of the maximum possible
  profits.
  
\item \emph{Viability analysis:} We propose a \emph{priority-based}
  pricing model that is general enough to be applicable for monetary
  profit maximization as well as social welfare maximization. We carry
  out mathematical analysis on the parameters to come up with
  necessary conditions for the viability of the feeder service and for
  routes with multiple trips to a hub node to be chosen in an optimal
  solution. Such analysis is useful for higher level planning.
  
\item \emph{Simulations:} We present a collection of simulations on
  multiple graphs that validate our analytical results. These
  simulation results include a realistic scenario of pre-positioning
  of relief material anticipating the devastation caused by a
  cyclone. We construct the graph and various other parameters in this
  simulation based on data from Odisha, a state in India that is prone
  to severe cyclones.
  
\item \emph{Comparisons with microscopic problem:} Through
  simulations, we also compare the macroscopic problem with the
  microscopic problem, wherein the demands, supplies and service
  pick-ups are all integers. We show that the macroscopic problem
  gives very good bounds on the optimal objective value for the
  microscopic problem. Moreover, the solution to the microscopic
  problem on the reduced route set, obtained from the proposed
  macroscopic analysis, is shown to be only slightly suboptimal
  compared to the solution of the microscopic problem on the original
  complete route set. The computational effort required to solve the
  microscopic problem on the reduced route set is significantly lesser
  compared to the effort required to solve the microscopic problem on
  the complete route set. Thus, the analysis on the macroscopic
  problem, and the resulting route reduction algorithm is very useful
  even for solving the more realistic microscopic problem.
\end{itemize}

On the modeling front, demand anticipative mobility problem
in~\cite{FR:2017,MS-etal:2018, MP-SLS-EF-DR:2012,RZ-MP:2016,
  AV-WG-SS:2017} though is for a macroscopic network flow model, it is
concerned with designing steady state flows whereas we consider
routing in a fixed hard time window. Similarly, ride sharing and taxi
dispatch do not address fixed horizon planning necessary for single
events, such as in disaster relief. Compared to the problems on
disaster relief~\cite{MA-etal-2015, MH-etal-2012, JMF-etal-2018,
  AMC-etal-2008, RL-ZJ-2019}, we present a unified model that is
applicable for both disaster response as well as non-humanitarian
applications such as transportation for large public events. Our model
can also incorporate priorities for demand at different nodes,
depending on say the local severity of a disaster.

Compared to the V.R.P. literature in general, including the papers on
disaster relief, we start with a macroscopic formulation and carry out
rigorous mathematical analysis on the optimal solutions. The resulting
route reduction algorithm and the demonstration of its computational
efficiency in solving even the microscopic problem is unique, to the
best of our knowledge. Similarly, the kind of viability analysis that
we do is typically missing in the V.R.P. and network flow literature.

In the preliminary version~\cite{SG-PT:2019-ecc} of this paper, we
considered only the feed-in problem and supply location problem. Here,
we additionally provide results for the feed-out problem and introduce
a model for carrying out the analysis on the viability of the feeder
service. Moreover, here we also provide the proofs missing
in~\cite{SG-PT:2019-ecc}. We also have an expanded simulation section
where we compare the macroscopic and microscopic versions of the
problem. Furthermore, we present a realistic application of planning
the relief distribution post a cyclone.

\subsubsection*{Organization} The rest of the paper is organized as
follows - we set up the \emph{one-shot} problem in
Section~\ref{sec:prob-setup}. In Section~\ref{sec:opt-sol} we describe
the properties of its optimal solutions and the \emph{off-line
  route-set reduction} method.  We present the priority-based pricing
model in Section~\ref{sec:pricing}. In Section~\ref{sec:max-profits},
we discuss the supply location problem and analyze maximum possible
profits with a given total supply for a given demand configuration. We
also show the equivalence with the supply optimized feed-out in this
section. We present the simulation results in
Section~\ref{sec:results} and provide concluding remarks in
Section~\ref{sec:conc}. We present the proofs of all but the main
theorems in appendices.

\subsubsection*{Notation}

We use $\integers$, $\nat$ and $\natz$ for the set of integers,
natural numbers and $\nat \cup \{0\}$, respectively. We use
$\intrangecc{a}{b}$ and $\intrangeoc{a}{b}$ to denote
$[a,b] \intersect \integers$ and $(a,b] \intersect \integers$,
respectively.

\section{Problem Setup} \label{sec:prob-setup}

In this section, we setup the coordinated first and last mode
transportation problem using a macroscopic formulation. Since this is
a network-flow problem, we first present a graph model, and then
define the route sets used for optimization. Finally, we present the
optimization model.

\subsection{Network Setup}
%


We assume that the operator wants to provide service in a region of
interest, which we model as a graph $G\ldef (V,E)$, where $V$ is the
set of nodes and $E$ is the set of edges between the nodes. Here, each
node $l \in V$ represents an area within the region of interest. For
both the feed-in and feed-out problems we have a special node, which
we call as the \emph{hub}, denoted by $H$. In the case of a feed-in
problem we have a demand flow from the areas of the region directed to
the hub $H$, while the opposite happens for the feed-out
problem. Thus, in both these problems, we assume that for each node
$l$ there is a demand $d_l \geq 0$ that needs to be serviced. This is
either spatially distributed in the region for the feed-in problem, or
is concentrated at the hub $H$ for the feed-out problem.  Furthermore,
we assume that the operator has a supply of $S_l\geq 0$ stationed at
each node $l$.

If there exists a direct connection from a node $l \in V$ to a node
$k \in V$, then we have a directed edge $(l,k)\in E$. Thus the
edge-set, $E$, represents the set of all direct connections within the
service region. Each edge has two weights $\rho_{lk}$ and $t_{lk}$,
which represent the per-unit flow travel cost and travel time,
respectively. We assume that the flows in our model do not affect the
congestion and thereby we assume that $\rho_{lk}, t_{lk}$ are not
affected by the flows we define later on.

We have the following assumptions for the supply configuration for
both the problems.
\begin{enumerate}[label=\textbf{(A\arabic*)},
  ref=\textbf{(A\arabic*)}] 
\item\label{A:1} For the feed-in problem, we assume that
  $\{S_l\}_{l \in V}$ is not necessarily optimized to service the
  demand.
\item\label{A:2}  The demand configuration $\{ d_l \}_{l \in V}$ is
fixed. Also, $d_l > 0$ for each node $l \neq H$ with $d_H = 0$, in both problems.
\item\label{A:3} 1 unit demand is serviced by 1 unit supply.

\end{enumerate} 
Assumption~\ref{A:1}, reflects a situation where the operator does not
have enough time to move the supply to demand-specific optimal
locations. These situations can arise specially in cases of
emergencies, where the situation demands immediate actions from the
operator. As an example, consider an emergency evacuation procedure in
case of an impending natural disaster. In~\ref{A:2}, we make the
assumption that for the nodes $l \neq H$ for which $d_l = 0$ purely
for ease of exposition and there is no loss of generality and does not
affect the analysis in any significant way. The assumption that
$d_H = 0$ is justified because the time and cost are `close' to zero
for travelling within the node $H$. Assumption~\ref{A:3} can be
justified by the fact that a unit of demand and a unit of supply can
refer to entities of different scales. As an example a unit demand can
be 100 parcels in a courier service whereas a unit supply can be a
single truck.

\subsection{Feasible Route Sets} \label{subsec:graph-mod}

We define a route $r \ldef (V_r, E_r)$ as a \emph{walk} in the graph
$G$. For a route $r$, $V_r$ is the sequence of nodes along the route,
with $V_r(j)$ being the $j^{\text{th}}$ node on the route $r$, whereas
$E_r(j)$ is the $j^{\text{th}}$ edge on the route $r$. Therefore,
\begin{align}
  &\map{ V_r }{ \intrangecc{1}{n_r} }{ V }, \ \text{s.t. }(
    V_r(j), V_r(j+1) ) \in E , \label{eq:nodeset}
  \\  & \map{ E_r }{ \intrangecc{1}{n_r-1} }{E }, \ \text{s.t. }
        E_r(j) = (V_r(j),V_r(j+1)),  \label{eq:edgeset}
\end{align}
where $n_r$ is the number of nodes (possibly repeated) on the route
$r$. We label the \emph{origin} and the \emph{destination} node of the
route $r$ as $o_r\in V$ and $D_r \in V$, respectively. That is
$o_r = V_r(1)$ and $D_r = V_r(n_r)$. We also define
$\displaystyle t_r \ldef \sum_{(l,k) \in E_r}t_{lk}$ as the
\emph{route travel time}.

We partition each route into \emph{legs}, which are sub-routes that
either originate and/or terminate at the hub $H$. We denote the
$i^\thh$ leg of route $r$ by $r^i = (V_r^i, E_r^i)$ with $V_r^i$ and
$E_r^i$ defined in the same manner as $V_r$ and $E_r$ respectively.
We classify all legs into two types. A \emph{circular leg} is one that
originates and terminates at the hub $H$, while a \emph{simple leg} is
one in which exactly one of the origin, $o_r$, or the destination,
$D_r$, is the hub. We identify the number of legs in a route $r$ by
$\theta_r \in \nat$. If $\theta_r = 1$ then we say $r$ is a
\emph{simple route}. We define a \emph{cycle} in a leg of a route as
any sub-sequence of nodes in $V_r^i$ which starts and ends at the same
node. Considering routes with multiple legs is particularly useful
when the supply is not capable of meeting the demand in one go, in
which case feeders can service the demand with multiple stops at the
hub $H$.

The operator has an objective of completing the trips of the feeder
vehicles within a hard time window of length $T$. Thus, we need to
identify the routes that are feasible for each specific problem under
consideration. We present here the feasible route sets for the
\emph{one-shot feed-in} and the \emph{one-shot feed-out} problems. In
the feed-in problem, the hub $H$ is the sink of all the demand
flows. Thus, a route $r$ is feasible in the feed-in problem if its
destination is $D_r = H$ and travel time $t_r \leq T$. We define the
feasible feed-in route set $\Rin$ as
\begin{align}\label{eq:rset}
  \Rin := \left\{ r\ |  \ D_r = H, \ t_r \leq T  \right\} .
\end{align}
Note that there exists a feasible route that passes through a node $l$
if and only if there is a feasible route $r$ with $o_r = l$.  In
general, there may be nodes through which no feasible route passes and
can be removed from $G$ without loss of generality.

For the feed-out problem the source of all demand 
is $H$. Therefore for a route $r$ to be feasible for the feed-out 
problem it 
should pass through the hub. Thus, the feasible feed-out route set 
$\Rout$ is
\begin{align}
  \Rout &\ldef \Rhat \cup \{r \vert o_r \neq H,\ t_r \leq T, \theta_r
          \geq 2\} , \label{eq:Rout}\\
  \Rhat&\ldef \{r \vert\ o_r = H, \ t_r \leq T \} . \label{eq:Rhat}
\end{align}
Here, $\Rhat$ denotes the set of all routes that originate from the
hub $H$ and have a route travel time less than or equal to $T$. Thus,
in the feed-out problem, we may have multi-origin multi-destination
routing. Since we want to consider both feed-in and feed-out problems
in a unified framework, we use the $\Mbf{R}$ as a general notation for
$\Rin$ or $\Rout$ whenever the discussion is not particular to only
one of the problems. Finally, we define $\srs{l, k, \Mbf{R} }$ as the
set of all simple routes $r \in \Mbf{R}$ that originate at node $l$
and terminate at $k$. Thus,
\begin{equation}
  \label{eq:simple-routes}
  \srs{l, k, \Mbf{R}} \ldef \{ r \in \Mbf{R} \ \vert \ o_r = l,\  D_r = 
  k, \ \theta_r = 1 \} .
\end{equation}

\subsection{Decision Variables and Constraints}
\label{subsec:constr}

We now introduce the decision variables and the constraints of the
problem. For each $r \in \Mbf{R}$, $f_r$ represents the \emph{volume
  of feeders taking a route $r$}.  Note that the route set $\Mbf{R}$
is an exhaustive set of all feasible routes given the time window
$T$. Therefore, if a route $r \in \Mbf{R}$ makes multiple visits to
the Hub $H$ then the route set $\Mbf{R}$ also contains a route
corresponding to every subroute of $r$ as well as each route formed by
retaining the first leg of route $r$ and permuting its secondary
legs. Thus, we assume that all the feeders in the flow $f_r$ traverse
the full route $r$. We let a \emph{service tuple} $\stup$ denote a
service on node $l$ in the $i^\thh$ leg of route $r$. We let $\fril$
denote the \emph{volume of demand the operator serves on $\stup$}. We
call $\fril$ \emph{service allocation} in short. In the feed-in
problem, $\fril$ is the demand that is picked up at $\stup$ while in
the feed-out problem it is the demand that is dropped off at
$\stup$. We let $F_l$ be the \emph{total demand served on a node
  $l$}. We say a service $\stup$ is infeasible if no demand can be
serviced at node $l$ on leg $i$ of route $r$. Even if a route
$r \in \Mbf{R}$ there may be some service tuples $\stup$ on the route
$r$ that may not be feasible. For example, consider a route not
originating from $H$ in the feed-out problem. The first leg does not
have any feasible services as demand can be picked-up only after the
vehicles first reach $H$. Thus, we let $\Mbf{W}(\Mbf{R})$ be the
\emph{set of feasible service tuples given a route set $\Mbf{R}$}.
\begin{equation}\label{eq:service-constr}
 \fril = 0,\quad  \forall \stup \notin \
\Mbf{W}(\Mbf{R}). 
\end{equation}
This implies that no service is possible on $\stup$ and thus all such tuples are not considered as decision variables. 


We now present the constraints on the decision variables, which are as
follows
\begin{subequations} \label{eq:constraints}
  \begin{align}
    &F_l \ldef  \sum_{r \vert l \in V_r} \sum_{i \vert l \in V_r^i}
      \fril \leq d_l, \quad \ \forall l \in V \label{eq:dem-constr}
    \\
    &\sum_{l \in V_r^i}\fril %
      \leq f_r, \quad \forall \ i \in \intrangecc{ 1 }{ \theta_r }, \
      \forall r \in \Mbf{R} \label{eq:legconstr}
    \\
    & \sum_{r \vert o_r = l} f_r %
      \leq S_l, \quad \forall \ l \in V .  \label{eq:supplyconstr}
  \end{align}
\end{subequations}
%
Constraint~\eqref{eq:dem-constr} ensures that the total demand served
at each node $l$ is at most the total demand at the
node. Constraint~\eqref{eq:legconstr} is the \emph{leg allocation
  constraint}, which ensures that the sum of all allocations in a leg
$i$ on a route $r$ is at most $f_r$, the feeder volume on that route,
while~\eqref{eq:supplyconstr} is the \emph{supply constraint}, which
ensures that the sum of feeder volumes on all routes originating from
node $l$, is at most the supply
$S_l$. Note that the flow conservation constraint is
automatically satisfied as volumes on each route are defined
separately. Hence we do not need to consider the constraint
explicitly.

\subsection{Optimization Model}

Next we give a model for the revenues and the cost to the operator and
then we summarize the overall optimization problem from the operator's
point of view. We let $\pril$ denote the price charged per-unit
allocation for service $\stup$. In general the prices $\pril$ are
dependent on many possible factors like service-times, demands
etc. However, we initially assume that the prices are independent of
the optimization variables $\fril$ and $f_r$ as well as demand and
supply variables $d_l$ and $S_l$. In Section~\ref{sec:pricing}, we
discuss the details of the pricing model and how it may be adapted for
different applications. Since the prices are known, the revenue over
all services is
\begin{equation*}
  \sum_{\stup}\pril \fril .
\end{equation*}

We consider two different types of costs incurred by the
service-provider. First, we let the travel cost for the volume of
vehicles that take the route $r$ be $c_r f_r$, which is the product of
travel cost per-unit flow and the volume of vehicles that go on route
$r$. The cost per-unit flow, $c_r > 0$, on a route $r$ is
\begin{align*}
  c_r &:= \sum_{(l,k) \in E_r} \rho_{lk} = \sum_{i = 1}^{\theta_r}
        \sum_{(l,k) \in E_r^i} \rho_{lk} \rdef \sum_{i = 1}^{\theta_r}
        c_r^i ,
\end{align*}
where $c_r^i > 0$ denotes the per-unit traversal cost on the
$i^{\thh}$ leg of route $r$. Second, we consider the \emph{operational
  costs} (which may include incentives or commissions to the drivers
and maintenance costs). We assume the operational cost is $k_l$ units
for every unit of allocation on a node $l$. The value of $k_l$ can be
positive or negative depending on the problem.

For each of the feed-in and feed-out problems with $f_r$ and $\fril$
as decision variables, we let the general \emph{one-shot operator
  profit maximization problem} be
\begin{equation}\label{eq:equivoptmodel}
  \begin{aligned}
    & \max_{f_r, \fril} J : = \sum_{\stup \in \Mbf{W}(\Mbf{R})} \bril
    \fril - \sum_{r\in \Mbf{R}}f_rc_r
    \\
    & \text{s.t. } \eqref{eq:constraints}, \ f_r ,\ \fril \geq 0 , \
    \forall \stup\in \Mbf{W}(\Mbf{R}) , \\ \text{ where, } &\Mbf{R} = \Rin, \text{ for
      feed-in problem } \\ &\Mbf{R} = \Rout, \text{ for feed-out
      problem } ,
  \end{aligned}
\end{equation}
and where $\bril$ is the per-unit \emph{operator income} for the
service $\stup$, which we define as
\begin{equation}\label{eq:Nodeprofitability}
  \bril \ldef \pril - k_l .
\end{equation}

\begin{remark}\label{rem:op-cost}\longthmtitle{Operational costs and profits}
  A positive value of $k_l$ is useful when the operator wants to
  reduce the incentives to the drivers or other maintenance costs. A
  negative value of $k_l$ models situations when the operator seeks to
  incentivize greater servicing of the demand or when the goal is to
  maximize social welfare rather than monetary profits. This latter
  case is useful in applications such as evacuation during natural
  disasters, wherein the operator needs to evacuate as many people as
  possible, which implies that the value $\sum_l(d_l -F_l)$ should be
  as low as possible. This is possible with $k_l<0$. Thus, we use the
  word profit to mean either the monetary profit or social profit or
  welfare.  \oprocend
\end{remark}


Thus, starting with the formulation~\eqref{eq:equivoptmodel} we solve
three problems in this paper. First we present an offline method for
reducing the size of the linear program, and thereby the computational
complexity. Then, we consider the \emph{feed-in supply optimization
  problem}, which extends~\eqref{eq:equivoptmodel} by considering the
supply configuration as an optimization variable in the feed-in
problem. Using this, we calculate the maximum profits for a given
demand configuration. We then show that a \emph{supply optimized
  feed-out} is equivalent to feed-in supply-optimization
problem. Finally, we present a simple model for generating the
parameters $\bril$ suitable for two major scenarios. In the first
scenario, the operator wants to maximize monetary profits. The second
scenario prioritizes the demand serviced over monetary costs. Thus, in
the latter scenario, the objective is to maximize a measure of social
welfare.


\section{ Properties of Optimal Solutions and Off-line Route
  Elimination} \label{sec:opt-sol}

In this section, we discuss some properties of the optimal solutions
of the problem~\eqref{eq:equivoptmodel}. With these properties we
reduce the size of the problem by eliminating routes in a given feasible
route set, $\Mbf{R}$, that would never be used in an optimal
solution irrespective of the demand and supply configurations.

\subsection{Properties of Optimal Solutions} \label{subsec:prop-opt-sol}

We first present necessary conditions for a route to have non-zero
allocations in an optimal solution of \emph{one-shot operator profit
  maximization problem}~\eqref{eq:equivoptmodel}. We derive these
necessary conditions from the K.K.T. conditions
for~\eqref{eq:equivoptmodel}.

\begin{proposition}\longthmtitle{Necessary conditions for a route to
    be used in an optimal solution}\label{prop:necc-cond-opt-sol}
  Let the route-set $\Mbf{R}$ be either $\Rin$ (feed-in) or $\Rout$
  (feed-out). In an optimal solution of the general one-shot
  problem~\eqref{eq:equivoptmodel}, if $f_r^* > 0$ for $r \in \Mbf{R}$
  then the following necessarily hold.
  \begin{enumerate}[label=(\alph*),
    ref=\ref{prop:necc-cond-opt-sol}(\alph*)]
		
  \item\label{itm:necc-a} $\fril^* > 0$ for some
    $i \in \intrangecc{1}{\theta_r}$ and $l \in V_r^i$. Further, for
    each $\stup$, if $\fril^*>0$ then $\beta_r^i(l) \geq 0$.

  \item \label{lem:legs} For a circular leg $i$ in $r$, the total
    service allocation is equal to the feeder volume on that route,
    $f_r$, that is,
    \begin{align}\label{eq:alloc-prop}
      \sum_{l \in V_r^i} \fril^*  = f_r^*, \quad \forall i \text{ such
      that } o_r^i = D_r^i= H .
    \end{align}
    
  \item\label{itm:necc-b} The route $r$ as a whole does not make a
    loss, that is,
    \begin{align*}
      \sum_{i = 1}^{\theta_r} \sum_{l \in V_r^i} \fril^* \bril \geq f_r^*
      c_r .
    \end{align*}
		
  \item\label{itm:necc-c} For each $i \in \intrangecc{2}{\theta_r}$
    and for $i = 1$ if leg $1$ of route $r$ is circular, there must
    exist an $l \in V_r^i$ such that $\fril^* > 0$ and
    $\bril \geq c_r^i$.
     
  \item\label{itm:necc-d} If $r$ is a simple route ($\theta_r = 1$)
    then there must exist atleast one $l \in V_r$ such that
    $f_r^1(l)^*>0$ and $\displaystyle \beta_r^1(l) \geq c_r$. \qed
	\end{enumerate}
\end{proposition}

Proposition~\ref{prop:necc-cond-opt-sol}, which we prove in
Appendix~\ref{prf:necc-cond-opt-sol}, states that irrespective of the
supply and demand configurations, if a route is used in an optimal
solution then the following must hold for that route.
\begin{itemize}
\item There is a positive service allocation that returns non-negative
  operator income.
\item For feeders moving in a circular leg, the optimal solution does
  not have an under-utilized feeder flow as it incurs a cost without
  earning any revenue.
\item The route, as a whole, does not make a loss.
\item Every secondary leg and the primary leg, if it is circular, of
  the route does not make a loss, that is, operator income from
  service allocations in such legs is no less than the per-unit
  traversal cost of the leg itself.
\item If the route is simple then it must have at least one node with
  non-negative per-unit operator income.
\end{itemize}

Using Proposition~\ref{prop:necc-cond-opt-sol}, we formulate an
off-line route reduction method in the next
subsection.

\subsection{Offline Route Elimination} \label{subsec:route-elim}

This subsection presents the \emph{reduced route set} for the feeding
problem~\eqref{eq:equivoptmodel} for a given route set $\mathbf{R}$,
which may be equal to either $\Rin$ or $\Rout$. In general, the
feasible route-set $\Mbf{R}$ grows rapidly with the time window
$T$. As a result, the problem size also grows accordingly. Thus, in
what follows, we obtain a reduced route-set that is formed by pruning
out routes and the corresponding optimization variables that would
have a zero allocation in every optimal solution under every possible
supply and demand configurations. We obtain this by application of the
individual properties in Proposition~\ref{prop:necc-cond-opt-sol},
after eliminating the dependence on $f_r$ and $\fril$.  We first
define the reduced route set $\Mc{R}(\Mbf{R})$, and then show that
every route $r$ such that $f_r^*>0$ in some optimal solution for some
supply and demand configuration to the feeding
problem~\eqref{eq:equivoptmodel} belongs to $\Mc{R}(\Mbf{R})$.
%
%
%
%
%
%
We define
\begin{subequations}
  \begin{align}
    \Mc{R}(\Mbf{R}) &\ldef \mathbf{R}_1 \cup \mathbf{R}_2 \label{eq:Rset}
    \\
    \mathbf{R}_1 &\ldef \left\{ r \in \mathbf{R} \ \vert \ \theta_r =
                   1, w_r^1 \geq 0\right\}\label{eq:R1}
    \\
    \mathbf{R}_2 &\ldef \{ r \in \mathbf{R} \ \vert \ \theta_r >
                   1, \sum_{i = 1}^{\theta_r} w_r^i \geq 0, \ 
                   w_r^i \geq 0, \forall i > 1 \label{eq:R2}
                   \} ,
  \end{align}
\end{subequations}

where
\begin{equation*}
  w_r^i \ldef \max_{l\in V_r^i}\{ \max\{ \bril, 0 \} \} - c_r^i .
\end{equation*}
Here, $w_r^i$ represents the maximum profits that can be earned per
unit service allocation in a leg $i$ of route $r$. Thus a leg of a
route is profitable on its own only if $w_r^i\geq 0$. Therefore,
$\Mc{R}(\Mbf{R})$ is the set of routes that are
profitable. Intuitively one expects that only the profitable routes to
be used in any optimal solution. In the following theorem we prove
this intuition formally.

\begin{theorem}\label{thm:algo-1}\longthmtitle{Optimal solutions to the feed-in problem
    use only the routes from the reduced route set} \label{th:routeprune}
	For the optimization problem~\eqref{eq:equivoptmodel}, every optimal
  solution for every demand and supply configuration is guaranteed to
  have $f_r^* = 0$ and consequently $\fril^* = 0$ over all legs $i$ of
  $r$, $\forall \ r \notin \Mc{R}(\Mbf{R})$ .
\end{theorem}

\begin{proof}
  We prove this result by contradiction - hence let there exist an
  optimal solution such that $f_r^* > 0$ for some
  $r \notin \Mc{R}(\Mbf{R})$. If $\theta_r =1$, then it satisfies
  Proposition~\ref{itm:necc-d} and as a consequence
  $\exists l \in V_r^1$ s.t.  $\beta_r^1(l) \geq c_r = c_r^1$, which
  implies $w_r^1\geq0$. Therefore $r\in \mathbf{R}_1$. Now, if
  $\theta_r>1$, then $r$ satisfies Proposition~\ref{itm:necc-c}. Hence
  \begin{equation*}
    w_r^i = \max_{l\in V_r^i}\{ \max\{ \bril, 0 \} \} - c_r^i \geq 0,
    \quad \forall i>1 .
  \end{equation*}
  Further, $r$ must also satisfy Propositions~\ref{itm:necc-a}
  and~\ref{itm:necc-b}, that is,
  \begin{align*}
    f_r^*c_r %
    &\leq \sum_{i = 1}^{\theta_r}  \fril^* \bril
      \leq \sum_{i = 1}^{\theta_r} \sum_{l \in V_r^i} \max\{ \bril, 0
      \} \fril^*
    \\ & \leq \sum_{i = 1}^{\theta_r} \max_{l \in V_r^i} \{\max
         \{\bril\}, 0 \}f_r^* ,
	\end{align*}
  where we have used the fact that $\fril^* > 0$ only if
  $\bril \geq 0$ for the second inequality and the third inequality
  follows from~\eqref{eq:legconstr}. Hence, $r$ must satisfy
  $\sum_{i = 1}^{\theta_r} w_r^i\geq 0$ which implies
  $r\in \mathbf{R}_2$. Therefore, in either case,
  $r \in \mathbf{R}_1\cup\mathbf{R}_2 = \Mc{R}(\Mbf{R})$.  This is a
  contradiction.
\end{proof}

In the proof we don't explicitly check Proposition~\ref{itm:necc-a}
for the route but one can verify that $\forall\ r\in \Mc{R}(\Mbf{R})$ there
exists $\bril \geq 0$ where $\fril^*>0 $ is possible for some supply
and demand configuration.  
 Thus, replacing $\Mbf{R}$ with 
$\Mc{R}(\Mbf{R})$ in Problem
\eqref{eq:equivoptmodel} causes no \emph{approximation} or \emph{loss} of any optimal
solutions. 

\section{Priority-Based Pricing Model}
\label{sec:pricing}

In this section, we propose a simple pricing model for the one-shot
problems based on the travel times and costs. This model helps in
further analysis of the one-shot model, like calculating maximum
profits presented in next section. We first present a general pricing
model and explore the question of viability of the feeder
service. Then, we specialize the general pricing model for two cases -
for maximization of serviced demand and for maximization of operator
profits. The former case is applicable in scenarios like evacuation or
relief supply at the time of disasters. The latter case is applicable
for large sporting and other public events as well as for one-shot
courier service.

\subsection{General Pricing Model}


%

We utilize the concept of \emph{value of time} (V.o.T.), which
associates a monetary cost to the travel times. In particular, we let
$\alpha$ be the monetary value of unit time. Then, the \emph{perceived
  cost} is $M + \alpha \tau$ for a transportation service that takes
$\tau$ units of time and charges a monetary price $M$. For each node
$l \in V$, we let $g_l \ldef \alpha \eta_l + \zeta_l$ be the perceived
cost for the \emph{best alternate transport}, which has a travel time
$\eta_l$ and has a price of $\zeta_l$. For the service $\stup$ to be
viable, the perceived cost of the feeder service should be less than
or equal to perceived cost for the best alternate transportation at
node $l$, that is
\begin{equation*}
\pril + \alpha\tril \leq \zeta_l + \alpha \eta_l\rdef g_l ,
\end{equation*}
where $ \tril$ is the \emph{service time} for the service tuple
$\stup$. This represents the duration of service including the waiting
time and the actual travel time for the service tuple $\stup$.

Referring to Remark~\ref{rem:op-cost}, note that since the motivation
of the operator is to maximize profits, either social or monetary,
%
%
the value of $\pril$ for any service $\stup$ has to be as
high as possible.  Thus, the optimal prices for the service $\stup$ is
\begin{equation}
\pril^*=  \zeta_l + \alpha \eta_l - \alpha\tril =
g_l - \alpha\tril . \label{eq:opt-pric}
\end{equation}

The next challenge is modelling of $g_l$, the perceived costs of the
best alternate transportation. We discuss this next, following which
we explore the question of economic viability of the feeder service.

\subsubsection{Modelling the Perceived Costs of the Best Alternate
  Transportation}

In order to systematically generate $g_l$ for each node $l \in V$, we
first assume that the best alternate transport available in the region
costs $b c_r$ and takes time $t_r$ along the route $r \in \Mbf{R}
$. The \emph{cost-factor} $b \geq 0$ signifies the cost to a passenger
of the best alternate transportation relative to the feeder
service. For simplicity, we assume it to be the same through out the
service area. We do this purely to simplify the notation and to carry
out some analysis regarding the economic viability of the feeder
service. The rest of the model and analysis in the paper can easily
handle more complicated perceived costs for the best alternate
transport.

First, we let
\begin{align}
 r(O,D,b)^* \in \argmin_{r \in \srs{O,D, \Mbf{R}}} \{\alpha t_r + b c_r\}\label{eq:best-alt-route}
\end{align}
where $r(O,D,b)^*$ is a route that the best alternate transport uses
between a given origin node $O$ and the destination node $D$. Note
that $\Mbf{R} = \Rin$ or $\Rout$ depending on the problem and we are
interested in setting one of $O$ or $D$ to the hub node $H$ and the
other to $l$, a node of interest. However, in order to avoid
cumbersome notation, we use the shortened notation of $r(l,b)^*$ as
defined in Table~\ref{table:1}.
\begin{table}[h]
  \begin{center}
    \begin{tabular}{| c | c | c | }
      \hline
      Shorthand Notation &\multicolumn{2}{|c|}{Complete Notation for Given $\Mbf{R}$}\\
      \cline{2-3}
                         &  $\Rin$ & $\Rout$  \\  
      \hline\hline
      (Origin,Destination) = $(O,D)$ & \((l,H)\) & $(H,l) $ \\ 
      \hline
      $ r(l,b)^* $ & $ r(l,H,b)^* $   &  $ r(H,l,b)^* $\\ 
      \hline
      $ c_{r(l,b)^*} $   & $ c_{r(l, H,b)^*} $ & $ c_{r(H,l,b)^*} $ \\
      \hline
      $ t_{r(l,b)^*} $ & $ t_{r(l, H,b)^*} $ & $ t_{r(H, l,b)^*} $ \\
      \hline
    \end{tabular}
  \end{center}
  \caption{Shorthand notations of route parameters for
    $\Mbf{R} = \Rin$ or $\Rout$}
  \label{table:1}
\end{table}
With this notation, we define the perceived cost of the best alternate
transport between nodes $l$ and $H$ given the route set $\Mbf{R}$ as
\begin{align}
g_l(b)&\ldef \alpha \eta_l + \zeta_l \numberthis , \quad \eta_l =
        t_{r(l, b)^*}, \quad \zeta_l = bc_{r(l,
        b)^*}. \label{eq:alt-perceived-cost}
\end{align}
Here, $\eta_l$, $\zeta_l$ and $g_l(b)$ are the travel time, cost and
the perceived cost of the best alternate transport between $l$ to node
$H$ with $(O,D)$ defined as per Table~\ref{table:1} based on
$\Mbf{R}$.

\begin{remark}\longthmtitle{Effect of the cost-factor on best
    alternate transportation}
  \label{rem:alt-transport-b}
  For a fixed $T$ and given $\Mbf{R} = \Rin$ or $\Rout$, as there are
  finitely many routes, the perceived cost $g_l(b)$ is a
  piecewise-linear, increasing, concave function of $b$ for each node
  $l \in V$. Further, at $b$, the slope of $g_l(b)$ is equal to
  $c_{r(l,b)^*}$ and the $g$ intercept is $t_{r(l,b)^*}$. Thus,
  $\eta_l$ and $\zeta_l$ are also unique for each $b$ except where the
  slope of $g_l(b)$ changes. For $b = 0$ and $b= \infty$, the routes
  $r(l,b)^*$ are the fastest and cheapest, respectively. For any $b$,
  $r(l,b)^*$ is such that $t_{r(l,b)^*} \geq t_{r(l,0)^*}$ and
  $c_{r(l,b)^*}\geq c_{r(l,\infty)^*}$. \oprocend
\end{remark}

\subsubsection{Viability of Feeder Service}

We first present conditions on the value of $b$ for the reduced route
set $\Mc{R}(\Mbf{R})$ to be non-empty and as a result for the feeder
service to be viable. An important factor for determining viability is
the operational cost/revenue $k_l$. As stated earlier, $k_l$
represents the cost that operator pays as commissions or represents
the incentive/subsidy for serving the node $l$. Thus, next we discuss
the viability of service depending on the sign of $k_l$.

\begin{lemma}\longthmtitle{Conditions on $b$ for viability of
    feeder service}\label{lem:viability}
  The reduced route set $\Mc{R}(\Mbf{R})$ is non-empty if and only if
  $g_l(b) \geq g_l(1)+k_l$, $l \in V$, where $ g_l(b) $ is given
  by~\eqref{eq:alt-perceived-cost}. Further,
  \begin{enumerate}[label = {(\alph*)},
  	ref=\ref{lem:viability}(\alph*)]
  	\item If $k_l > 0$ $\forall l \in V$ then $\Mc{R}(\Mbf{R})$ is non-empty only if $b > 1$.
  	\item If $k_l<0$, for all $l \neq H$ and $b= 1$, then
          $\Mc{R}(\Mbf{R})\neq \emptyset$.
  \end{enumerate}  
\end{lemma}

We present the proof in Appendix~\ref{prf:viability}. Next, we present
necessary conditions for the viability of \emph{multi-service
  routes}. These are routes with multiple services over different
legs, viability of which helps service more demand with lesser supply.

\begin{proposition}\label{prop:multi-legged} \longthmtitle{Necessary
    value of $b$ for \emph{multi-service viability}}
  Consider the graph $G$, the route set
  $\Mbf{R} \in \{ \Rin, \Rout \}$ and the corresponding $(O,D)$
  defined in Table~\ref{table:1}. Suppose $k_H > 0$ and let $c^*(p,q)$
  be the cheapest travel cost from $p$ to $q$ for route set
  $\Mbf{R}$. Suppose that $\exists (r, i_1, l_1) \in \Mbf{W}(\Mbf{R})$
  and $\exists (r, i_2, l_2) \in \Mbf{W}(\Mbf{R})$ for a route
  $r\in \Mc{R}(\Mbf{R})$, two distinct legs,
  $i_1, i_2 \in \intrangecc{1}{\theta_r}$ and $l_1, l_2 \in V_r$,
  where $\Mbf{W}(\Mbf{R})$ is defined as
  per~\eqref{eq:service-constr}. Then,
  \begin{enumerate}[label = {(\alph*)},
    ref=\ref{prop:multi-legged}(\alph*)]
  \item \label{prop:necc-suff-b-multi-legs}
    $g_l(b)\geq g_l(1)+ c^*(D,O)+k_l$, for some $l \in V$, $l\neq H$
  \item \label{prop:necc-b-multi-legs}
    $\displaystyle b \geq \hat{b}_l^*$ for some $l \in V$, $l \neq H$,
    where
    $\displaystyle \hat{b}_l^* \ldef \left( \frac{k_l+ c^*(D,O)+
        \alpha(t_{r(l,1)^*}-t_{r(l,\infty)^*})+ c_{r(l,1)^*} }{
        c_{r(l,0)^*}} \right)$
  \item \label{prop:necc-b-multi-legs-kl>0}%
    If $k_l \geq 0$ for all $l \in V$ then $\exists l \in V$,
    $l \neq H$ such that $\displaystyle b\geq b_l^*$, where
    $b_l^* \ldef \left(1 + \frac{k_l+ c^*(D,O) +
        \alpha(t_{r(l,1)^*}-t_{r(l,\infty)^*})}{ c_{r(l,1)^*}}
    \right)$.
  \end{enumerate}
  Also, $b$ calculated by (a) is greater than $\hat{b}_l^*$. Also,
  $b_l^* \geq \hat{b}_l^*$.
\end{proposition}
The proof is given in Appendix~\ref{prf:multi-legged}. Note that the
assumption that $k_H > 0$ causes no loss of generality in our
problem. This is because in both the feed-in and feed-out problems, we
never have demand that wants to go from $H$ to $H$. So, we can choose
$k_H > 0$, without affecting any optimization variables.

With Proposition~\eqref{prop:multi-legged}, given the graph $G$, the
value of time $\alpha$ and the value of $b$, the operator can evaluate
\emph{multi-service} viability of the feeder service. Also, given a
node $l$, Proposition~\ref{prop:necc-suff-b-multi-legs} can be
interpreted as a necessary condition for viability of a single depot
V.R.P. service.  Note that we give three necessary conditions in
Proposition~\ref{prop:multi-legged}. The bound in claim~(a) requires
the computation of $g_l(b)$ for each $b$. In comparison, condition~(b)
is a computationally simpler relation but provides a bound lower than
the one in condition~(a). Condition~(c) provides a stricter bound than
condition~(b) under the assumption that $k_l>0$ for all $l \in
V$.

\begin{remark}\label{rem:reduced-routes}
  \longthmtitle{Reduced route set reaches a constant as the
    destination time $T$ is increased} %
  There exists a time $T^*$ such that for all $T \geq T^*$, the set
  $\Mc{R}(\Mbf{R})$ is the same. This is because even though the set
  $\Mbf{R}$ includes more and more routes as $T$ increases, for a
  route to be in $\Mc{R}(\Mbf{R})$, the service times $\tril$ cannot
  exceed a certain value while satisfying $\bril \geq 0$. Furthermore,
  even for services with lower service times, higher costs would
  render longer routes unprofitable. Thus $\Mc{R}(\Mbf{R})$ do not
  have such routes. This is particularly useful as the set $\Mbf{R}$
  and problem~\eqref{eq:equivoptmodel} keeps growing with $T$, whereas
  the size of~\eqref{eq:equivoptmodel} with $\Mc{R}(\Mbf{R})$ instead
  of $\Mbf{R}$ does not grow for $T \geq T^*$. \oprocend
\end{remark}
 
\subsection{Pricing Models for Different
  Scenarios} \label{sec:pricing-scenarios}

We next demonstrate the generality of our model and explain the
interpretations of $\bril$, $\tril$, $k_l$ and thereby the objective
of~\eqref{eq:equivoptmodel} for two broad scenarios.

\subsubsection{Demand Maximization}

In the first scenario, the operator aims to maximize the social
welfare by maximizing the demand that is serviced. This case is
applicable in disaster management scenarios, where the goal is to
maximize the evacuation (feed-in) or relief efforts
(feed-out). 
Thus in this scenario, one could let the service time $\tril$ for
$\stup$ be the drop-off time with respect to a suitable reference
initial time. In other words, the service time can include both the
time to wait for the feeders to pick-up as well as the travel time
from node $l$ to the drop-off location on leg $i$ of route $r$.

Next, $k_l$ in this scenario can represent a sense of priority with
which a region $l$ must be serviced. Referring to
Remark~\ref{rem:op-cost}, we thus let $k_l$'s to be negative. For
example, regions with greater danger of loss of life could have a more
negative value of $k_l$. Alternatively, $k_l$ could also represent
monetary incentives or subsidies that the operator would receive from,
say the government, for servicing the region $l$. The ``operator
income'' $\bril$ in this scenario may represent the ``social revenues"
rather than monetary income for a service $\stup$. Alternatively,
$\bril$ could also represent monetary income if $k_l$ represents
monetary incentives or subsidies that are given to the operator.


Thus in this scenario, the objective of~\eqref{eq:equivoptmodel} is to
maximize social welfare in a cost efficient and coordinated
manner. Furthermore, Lemma~\ref{lem:viability} and
Proposition~\ref{prop:multi-legged} provide conditions when relief
efforts are not viable either due to lack of time before the disaster
or due to the costs.

\subsubsection{Profit Maximization}

In this scenario, the operator seeks to maximize the monetary
profits. Thus, $\tril$ may or may not include wait times depending on
the specific instance of the problem. The parameter $k_l$ could be the
fixed costs that the operator needs to pay for all the expenses such
as maintenance of the fleet, commissions to the drivers etc. In this
scenario, it is also possible that a region $l$ may pay incentives for
availing service from the operator. Therefore, $k_l$ can potentially
be negative though in general it is positive for this scenario. The
objective in~\eqref{eq:equivoptmodel} therefore represents the
monetary profits. This scenario is applicable for coordinated
transportation services in the context of less frequent events like
long distance courier services, where holding at the main airport hub
is costly and parcels need to arrive and depart just in time; or large
social gatherings or sporting events.

\section{Feed-in Supply Optimization Problem and Equivalence to One-Shot Feed-out}
\label{sec:max-profits}

Until the previous section, we discussed the one-shot feeding problem,
given a supply and demand configuration. However, an operator may also
be interested in determining the best possible supply configuration
for a demand configuration so that profits or social welfare are
maximized. In the scenarios where evacuation or relief efforts in the
context of a disaster are the goal, predictive models of the impact of
a disaster may be available. This information could be used to
pre-position the service vehicles well before the evacuation or relief
distribution begins. In this section, we first discuss the problem of
supply optimization for feed-in. Then, we discuss supply optimization
for feed-out and show its ``equivalence'' to supply optimized feed-in
under mild assumptions.

We assume that the demand configuration and the total available supply
$s$ are given. We let the supply at the nodes $\{S_l\}_{l \in V}$ also
be optimization variables. Then, \emph{feed-in supply optimization
  problem} can be stated as
\begin{align}
  & \max_{S_l, f_r, \fril} \bar{J} : = \sum_{(r,i,l)} \bril \fril - \sum_{r \in
    \Rin}f_rc_r \notag
  \\
  & \text{s.t. }  \eqref{eq:constraints}, \ \sum_{l \in V} S_l \leq s,
    \ \forall\ l \in V_r \notag
  \\
  &f_r ,\ \fril, \ S_l \geq 0, \ \forall r \in \Rin, \
    \forall i \in \intrangecc{1}{\theta_r}, \ \forall l \in V_r
    .  \label{eq:supply-opt}
\end{align}

For analysing the maximum profit as a function of the total supply
$s$, we make the following assumptions.%
\begin{enumerate}[label=\textbf{(A\arabic*)},
  ref=\textbf{(A\arabic*)}]
  \setcounter{enumi}{3}
\item \label{A:4} For each node $l$ in the graph $\exists\ r \in \Rin$
  s.t. $o_r =l$, $\theta_r =1$ and $\beta_r^1(l) - c_r^1 \geq 0$.
\end{enumerate}
Note that there is no loss of generality in Assumption~\ref{A:4}. We
elaborate on this point in the sequel.

\subsection{General Properties of Feed-in Supply
  Optimization} \label{subsec:supply-optim-prop}

In the following proposition, we give some properties of the optimal
solutions of the \emph{feed-in supply optimization problem}
\eqref{eq:supply-opt}. From the result, we can also reason that there
is no loss of generality in the Assumption~\ref{A:4}. We present the
proof of the result in Appendix~\ref{prf:supply-opt}.

\begin{proposition}\longthmtitle{Properties of optimal supply
    configurations and allocations}
  \label{prop:supply-opt} 
  In every optimal solution to the problem~\eqref{eq:supply-opt}, the
  following hold:
  \begin{enumerate}[label=(\alph*), ref=\ref{prop:supply-opt}(\alph*)]
  \item \label{itm:sup0} If $f_r^* >0$ then $r\in \MbrR$

  \item \label{itm:sup1} For each $r \in \MbrR$,
    $f_r^1(o_r)^* = f_r^*$. Consequently,
    \begin{align}\label{eq:no-redund-flow}
      \sum_{l \in V_r^i} f_r^i(l)^* = f_r^*, \quad \forall i
      \in \intrangecc{1}{\theta_r}, \quad \forall r \in \MbrR ,
    \end{align}
    and no route originating from $H$ is used.
	
  \item \label{itm:sup2} For a route $r\in \MbrR$ if $f_r^*>0$ then
    $\beta_r^1(o_r)\geq c_r^1$. Consequently,
    $\forall i \in \intrangecc{1}{\theta_r}$, $\exists\ l \in V_r^i$
    such that $\fril^*>0 $. Moreover, for $\stup$, $\fril^* > 0$ only
    if $\bril \geq c_r^i$.
  
  \item \label{itm:sup3} If a node $l$ does not satisfy the property
    in~\ref{A:4}  then $\fril^* = 0$ for all $\stup$ that serve the
    node $l$.
%
%

  \item \label{itm:sup4} A route $r$ with a cycle in the first leg is
    not used, that is $f_r^* = 0$. 
  
  \item \label{itm:sup5} If $\displaystyle s\leq \sum_{l \in V} d_l$
    and \ref{A:4}
    holds then
    $\displaystyle \sum_{r \vert o_r = l} f_r^* = S_l \leq d_l, \
    \forall l\in V$. \qed
  \end{enumerate}
\end{proposition}

From Proposition~\ref{prop:supply-opt} we see that there is no loss of
generality in the Assumption~\ref{A:4}. This is because
if for a node $l$, $d_l = 0$ then, by Proposition~\ref{itm:sup1}, all
routes $r$ originating at $l$ have zero flow ($f_r = 0$) in every
optimal solution. Similarly, Proposition~\ref{itm:sup3} says that in
every optimal solution there is no allocation on nodes that violate
Assumptions~\ref{A:4}. With Proposition~\ref{itm:sup3} one can
eliminate the nodes that do not follow assumption \ref{A:4}.
%
%
Also, as a consequence of Propositions~\ref{itm:sup1} and
~\ref{itm:sup2}, one can eliminate the route flow variables $f_r$ and
remove routes without a \emph{profitable} pickup at their origins.
With Proposition~\ref{itm:sup4} we can eliminate every route with a
cycle in the first leg. Thus we construct the \emph{reduced route set}
for~\eqref{eq:supply-opt}, $\redR$, as
\begin{equation}\label{eq:reduced-fiso-rset}
  \redR \ldef \{r \in \MbrR\vert  o_r \neq H, \beta_r^1(o_r)\geq
  c_r^1, \text{ no cycles in } r^1   \} .
\end{equation}

Further, using Proposition~\eqref{prop:supply-opt}
and~\eqref{eq:no-redund-flow}, we can solve~\eqref{eq:supply-opt} with
equality in the constraints of~\eqref{eq:legconstr}
and~\eqref{eq:supplyconstr}. Thus, we can reduce~\eqref{eq:supply-opt}
to an optimization problem over decision variables $\fril$, the
allocations, and $S_l$, the supply at a node. This elimination of the
variables $f_r$ leads to a significant reduction in the number of
optimization variables, specifically equal to the number of routes in
$\redR$. As a result, we can express the supply optimization problem
as
\begin{align}\label{eq:max-profits}
  &\max_{S_l,\ \fril } \bar{J} = \sum_{\stup\in \Mbf{W}(\redR)} (\bril-c_r^i)\fril
  \\
  & \text{s.t. }  \ F_l \leq d_l, \ \sum_{l \in
    V_r^i}\fril = f_r, \ \sum_{r \vert o_r = l} f_r = S_l, \ \sum_{l \in V}
    S_l \leq s, \notag
  \\
  &\fril, \ S_l \geq 0, \ \forall \stup \in \Mbf{W}(\redR) . \notag
\end{align}
This is a simpler problem to solve for a sequence of values of $s$
than \eqref{eq:supply-opt}. Also, as we show in the next subsection,
this formulation makes the feed-out problem computationally simpler.

\subsubsection{Absolute Maximum Profits} \label{subsec:suf-suppl}

With the objective \eqref{eq:max-profits}, we can also analyze the
absolute maximum profits an operator can earn, over all supply
configurations, for a given demand configuration. Quantification of
the absolute maximum profits is useful for determining profitability
of the service from the operator's perspective. To arrive at the value
of \emph{absolute maximum profits}, denoted by $J_{\max}$ from here
on, we assume (without loss of generality) that supply is sufficient,
i.e.  $s\geq \sum_l d_l$.
%
%
%
We denote the set of simple routes with the
maximum rate of profits for a pickup at $l$ by $\Mc{S}(l)$. Formally,
\begin{equation} \label{eq:best-routes}
\Mc{S}(l) \ldef \argmax_{r \in \srs{l, H, \Rin} \cap \redR}
\{\beta_r^1(l)-c_r \} .
\end{equation}
\begin{lemma} \label{lem:Rl-least-perceived-cost} For each node
	$l \in V$, $\beta_r^1(l)-c_r > \beta_{\bar{r}}^i(l)-c_{\bar{r}}^i$
	for all $r \in \Mc{S}(l)$, $\bar{r} \notin \Mc{S}(l)$ and
	$i \in \intrangecc{1}{\theta_{\bar{r}}}$. Further,
	$\forall r \in \Mc{S}(l)$ the perceived cost $(\alpha t_r +c_r)$ is
	the least from node $l$. \qed
\end{lemma}
The lemma is proved in Appendix \ref{prf:Rl-least-perceived-cost}.
This lemma implies that the pickups in the simple routes with maximum
rate of profits are always more profitable than those of any other
route. Next, we state the properties of optimal solutions
of~\eqref{eq:max-profits} for sufficient supply.

\begin{theorem}\label{th:suff-supplies}
  \longthmtitle{Properties of optimizers under sufficient supply} %
  If Assumption~\ref{A:4} holds then the following hold.
  \begin{enumerate}[label = {(\alph*)},
    ref=\ref{th:suff-supplies}(\alph*)]

  \item \label{itm:thm6a} If $s\geq \sum d_l$ then $f_r^* =
  0$,  $\forall r \notin \Mc{S}(o_r)$ . Further, for each $l \in V\setminus \{H\}$,
    $\displaystyle F_l^* = \sum_{r \in \Mc{S}(l) } f_r^* = d_l$ and
    $S_l^* \geq d_l $.
    \item \label{lem:supply-set}
    For any total supply $s$, the set of supply configurations defined in~\eqref{eq:supply-set} always contains an optimal configuration.
    \begin{equation}\label{eq:supply-set}
    \left\{\{S_l\} \Big\vert \sum_{l \in V} S_l = s,  \ S_H = \max\{0, s- \sum_{l \neq
    	H} d_l \} \right\} 
    \end{equation}
  \item \label{itm:thm6c} The maximum profits over all supply
    configurations is
    \begin{align}\label{eq:Jmax}
      J_{max} = \sum_{l\in V} d_l \max_{r \in \Mc{S}(l)}
      \{\beta_r^1(l)-c_r^1 \} .
    \end{align}
  \end{enumerate}
\end{theorem}

\begin{proof}
  \textbf{(a):} As $s \geq \sum_l d_l$, consider a solution where,
  $\displaystyle \sum_{r \in \Mc{S}(l)}f_r^1(l) = d_l$, $\fril = 0$
  for all $r \notin \Mc{S}(l)$ for each $l \in V$, and
  $\displaystyle S_l = \sum_{r \in \Mc{S}(l)}f_r^1(l)$,
  $\forall\ l\neq H$, and $\displaystyle S_H = s - \sum_l d_l$. One
  can verify that such a solution is feasible under
  Assumption~\ref{A:4}. From Lemma~\ref{lem:Rl-least-perceived-cost},
  we know that $\beta_r^1(l)-c_r > \beta_q^i(l)-c_{q}^i$ for all
  $r \in \Mc{S}(l)$ and $q \notin \Mc{S}(l)$. Then the structure of
  the objective function~\eqref{eq:max-profits} implies that this
  solution is optimal.
  
   \textbf{(b):} Part~(b) follows from arguments given above and Proposition~\ref{itm:sup5}.
   
  \textbf{(c):} Given part~(a), we now see that the maximum profits
  must satisfy~\eqref{eq:Jmax}, in which the term indexed by $l$
  corresponds to the profits from node $l$.
\end{proof}

This theorem gives the absolute maximum profits for a given demand configuration over all supply
configurations. The value of
$J_{\max}$ is easily computable with knowledge of maximum rates of profit for simple routes and  demand configuration.

\subsection{Supply Optimized Feed-out}

In this subsection, we explore \emph{supply optimized feed-out
  problem}. We denote the variables and parameters of this problem
with a hat. Thus $\hfr$, $\hfril$ and $\hbril$ $\hat{S}_l$ represent
flow, allocations and operator income respectively. Thus with demand
configuration $\{\hdl\}$, and optimization variables $\hfr$, $\hfril$,
total allocation $\hFl$ and supply configuration $\{\hat{S}_l\}$, the
problem is
\begin{equation}\label{eq:supply-opt-fo}
\begin{aligned}
& \max_{\hfr, \hfril, S_l} J : = \sum_{\stup \in \Mbf{W}(\Rout)} \hbril
\hfril - \sum_{r\in \Rout}\hfr c_r
\\
& \text{s.t. } \eqref{eq:constraints},\quad  \sum_l \hat{S}_l \leq s , \quad \hat{S}_l \geq0 , \ \forall l \in V
\\ &\quad \hfr ,\ \hfril\geq 0 , \forall \stup\in \Mbf{W}(\Rout) .
\end{aligned}
\end{equation}
It is trivial to note that the optimal supply configuration in the
case of feed-out is to concentrate the supply at the hub $H$. Thus, we
assume that $\hat{S}_H = s$, for total supply $s$. Furthermore, all
the optimal routes lie in $\Rhat$ (see~\eqref{eq:Rhat}). With this
assumption we show that for the supply optimized feed-out problem, an
equivalent feed-in supply optimization problem exists with the same
optimization value and related optimizers.
%
Thus, for the feed-out problem, we assume a graph
$\hat{G} =(\hat{V},\hat{E})$ and construct a solution given through
the following conditions.
\begin{enumerate}[label=\textbf{(C\arabic*)}, ref=\textbf{(C\arabic*)}]
  \setcounter{enumi}{0}
\item\label{C:supply-H} For the feed-out problem we assume that
  $\{\hat{S}_l\}$ is such that $\hat{S}_H = s$ and $\hat{S}_l =0$.
\item \label{C:reverse-graph} Let the graph $G \ldef (V, E)$ be such
  that $V \ldef \hat{V}$, edge $(k,l) \in E$ iff $(l,k) \in \hat{E}$
  and $(\rho_{kl}, t_{kl}) = (\hat{\rho}_{lk}, \hat{t}_{lk})$
  $\forall (l,k) \in \hat{E}$. Also, we assume $T = \hat{T}$.
\item \label{C:reverse-demand} The supply configuration for
  \eqref{eq:supply-opt} is chosen from the supply set
  \eqref{eq:supply-set}. Also, spatial demand configurations are same,
  i.e. $\{\hdl\}=\{d_l\} $.
\end{enumerate}

Condition~\ref{C:supply-H} requires the supply to be concentrated at
$H$, which as we stated earlier is optimal for the feed-out problem.
Condition~\ref{C:reverse-graph} represents a ``reversed" graph for the
feed-out problem with same parameters. Similarly,
Condition~\ref{C:reverse-demand} ``reverses'' the demand. Next, we
give a mapping between the route set $\Rin$ and $\Rhat$.
  
\begin{remark}\longthmtitle{Equivalent
    Route-Sets} \label{rem:equiv-routes} Given a route $r$ on graph
  $\hat{G}$, let $ \rev \ldef \bar{r}$, a route in $G$ (defined by
  \ref{C:reverse-graph}) where $\theta_r = \theta_{\bar{r}}$,
  $V_r(i)= V_{\bar{r}}(n_{\bar{r}}-i+1)$ and
  $E_{\bar{r}}(n_{\bar{r}}-i)= (V_r(i+1),V_r(i)) $. Then,
  $\forall r \in \Rhat$, $\exists \rev = \bar{r}\in \Rin$ with
  $c_r = c_{\bar{r}}$ and every service tuple $\stup$ of $r$ mapping
  to $(\bar{r}, \theta_r-i+1, l)$. \oprocend
\end{remark}

We can now show that $\hbril$ and
$\beta_{\bar{r}}^{\theta_r-i+1}(l)$ are the same under
Condition~\ref{C:reverse-graph}.
\begin{lemma}\label{lem:equal_prices}
	For the feed-out problem and the corresponding feed-in problem,
	$\hbril= \beta_{\bar{r}}^{\theta_r-i+1}(l)$, where $\bar{r} = \rev$.\qed
\end{lemma}
The proof of this Lemma is given in Appendix~\ref{prf:equal_prices}. 
 Next we show equivalence of the two problems. 

\begin{theorem}\longthmtitle{Equivalence of the feed-out problem and the
    feed-in supply optimization problem} \label{prop:equivalence}
  Under Assumption~\ref{A:4}, the feed-out problem defined by route
  set $\Rhat$ with supply concentrated at the Hub can be represented
  by an equivalent feed-in supply optimization
  problem~\eqref{eq:supply-opt} (which can be constructed as in
  Conditions~\ref{C:supply-H}-~\ref{C:reverse-demand}) with
  $\hat{J}^*(s) = \bar{J}^*(s)$, $\forall s \geq 0$. Further, for all
  optimal solutions
	\begin{enumerate}[label=(\alph*),
		ref=\ref{prop:equivalence}(\alph*)]
  \item\label{itm:equiv-b}
    $\hfril^* = f_{\bar{r}}^{\theta_r-i+1}(l)^* $, $\forall \stup$ and
    $\bar{r} = \rev$,
  \item\label{itm:equiv-c} $\hfr^* = f_{\bar{r}}^*$, $\bar{r} = \rev$,
    $\hFl^* = F_l^*$, $\forall\ l \in V$.
	\end{enumerate}
\end{theorem} 

\begin{proof}
	Each pair of $r\in \Rhat $ and $\rev = \bar{r} \in \Rin$
  satisfy the relationship given in Remark \ref{rem:equiv-routes}.
  Using Lemma~\ref{lem:equal_prices}, we see that
  $\hbril = \beta_{\bar{r}}^{\theta_r - i+1}(l)$ and given
  $T = \hat{T}$ we conclude that the cost functions are equivalent,
  i.e.  $\bar{J} \equiv \hat{J}$ under the assumption
  $\hfril = f_{\bar{r}}^{\theta_r-i+1}(l)$. Hence, it is sufficient to
  prove equivalence of the constraints in both problems.
	
	\emph{Leg Constraints}: In both problems given flows $\hfr$ and
  $f_{\bar{r}}$ the constraint for leg $i$ and route $r$ in
  \emph{feed-out} problem and the constraint
  for leg $\theta_r -i+1$ and route $\bar{r}$ for the feed-in supply
  optimization (see \eqref{eq:max-profits}) are equivalent.
	
	\emph{Demand Constraints}: Constraint $F_l\leq d_l $ are equivalent under the assumption that $d_l= \hdl$.
	
	\emph{Supply Constraints}: To show this equivalence, we let
  $\{\hat{S}_l\}_l $ be the supply configuration at the end of
  \emph{feed-out}. We know that for any route $r$, the flow terminates
  at $D_r$. Therefore, the final supply located at any node $l$ is
  $\displaystyle \hat{S}_l = \sum_{r \vert D_r= l} \hfr$. Therefore,
  Constraint~\eqref{eq:supplyconstr} for feed-out can be rewritten as
  \begin{align*}
    \sum_r\hfr = \sum_{l}  \sum_{r\vert D_r = l }\hfr = \sum_l
    \hat{S}_l \leq s ,
  \end{align*}
  which implies,
  $ \hat{S}_l = \sum_{r \vert D_r= l} \hfr, \text{ and} \quad \sum_l
  \hat{S}_l \leq s $.  Furthermore, one can see that
  $\sum_r f_r = \max\{\sum d_l, s\}$ which implies that all the supply
  reaches $H$ in feed-in problem. Now, with the assumption that supply
  configuration for \emph{feed-in} is chosen from the
  set~\eqref{eq:supply-set}, one can see that constraint
  \eqref{eq:supplyconstr} along with $\sum_{l \in V} S_l \leq s$ are
  equivalent to ones stated above as consequences of Proposition
  \ref{itm:sup5} for $s\leq \sum_l d_l$ and Theorem \ref{itm:thm6a}
  with Lemma \ref{lem:supply-set} for $s \geq \sum_l d_l$.
\end{proof}

Theorem~\ref{prop:equivalence} establishes the equivalence of the
feed-out problem on the graph $\hat{G}$ to the supply optimization
problem on the graph $G$, formed using \ref{C:reverse-graph}. Thus,
one may solve either problem and obtain a solution to the feed-out
problem. More importantly, all the properties and results of supply
optimization problem apply for the supply optimized feed-out problem,
though in a ``reversed'' manner.

\section{Simulations}\label{sec:results}

Since the general operator feeding problem~\eqref{eq:equivoptmodel} is
a linear program, we utilized CVXpy~\cite{SD-SB:2016} for modelling
the problem. We also performed simulations for a microscopic
formulation of the problem~\eqref{eq:equivoptmodel}, with integer
constraints on the optimization variables. We performed these
microscopic simulations using Gurobi~\cite{gurobi} and ECOS
optimizers. We performed all the simulations on a 3.7 GHz
i7-7$^{\thh}$ generation machine with 32 GB of RAM. We present
simulation results for several graphs that include test graphs and a
graph constructed from real world data. With these simulation results,
we illustrate various aspects of our analytical results and
comparisons with solutions of the problem in a microscopic
formulation.

\subsection{Numerical Illustration of Analytical Results}

\subsubsection{Test Graphs Utilized for Simulation}

We have utilized the two test graphs in Figure~\ref{fig:graphs} for
numerical illustration of our analytical results. The graph in
Figure~\ref{fig:20_node} has undirected or bidirected edges and edge
weights (travel time and travel cost). Some highway networks (such as
in semi-urban or rural areas) and railway networks are of this
nature. The graph in Figure~\ref{fig:23_node} has directed arcs while
the travel times and costs depend on the direction of the arc between
each pair of nodes. Such conditions are common for transportation
networks in cities.
\begin{figure}[h]
	\begin{subfigure}[b]{0.5\textwidth}
		\centering
		\includegraphics[width=\columnwidth ]{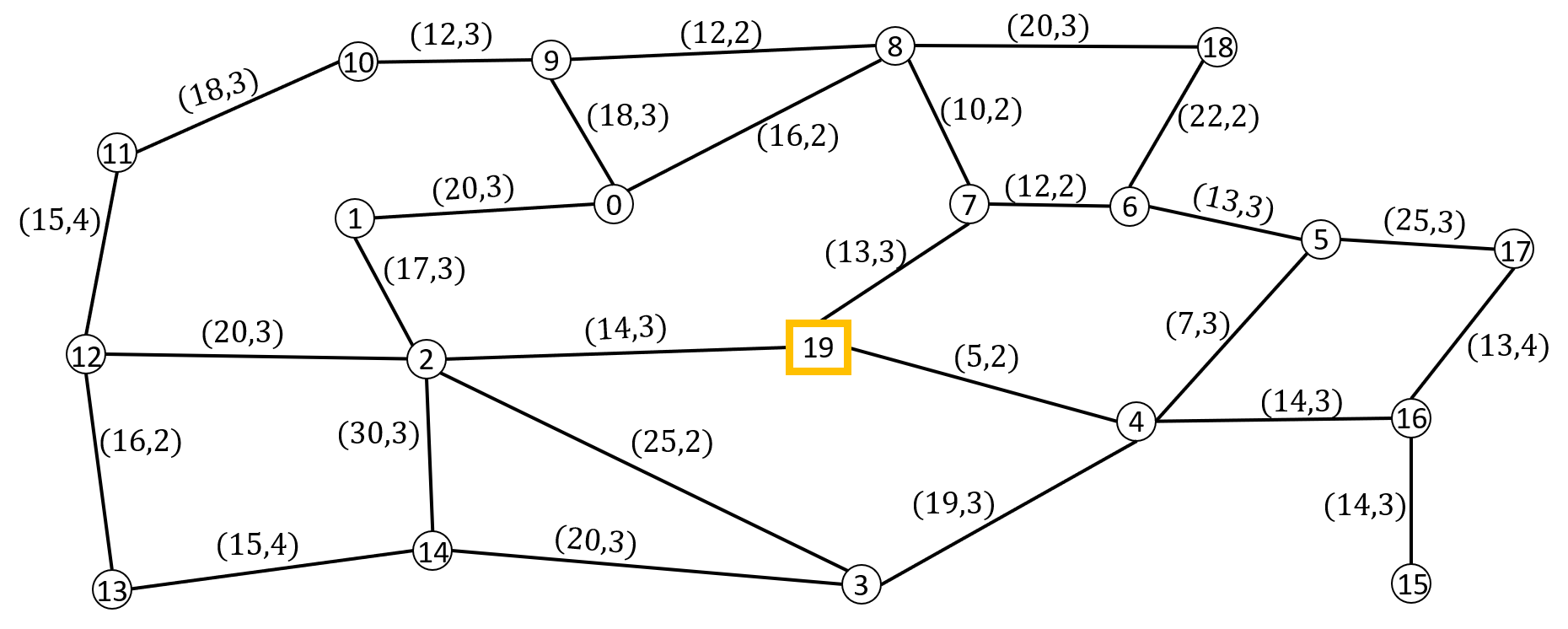}
		\caption{}\label{fig:20_node}
	\end{subfigure}%
	\begin{subfigure}[b]{0.5\textwidth}
	\centering
	\includegraphics[width=\columnwidth ]{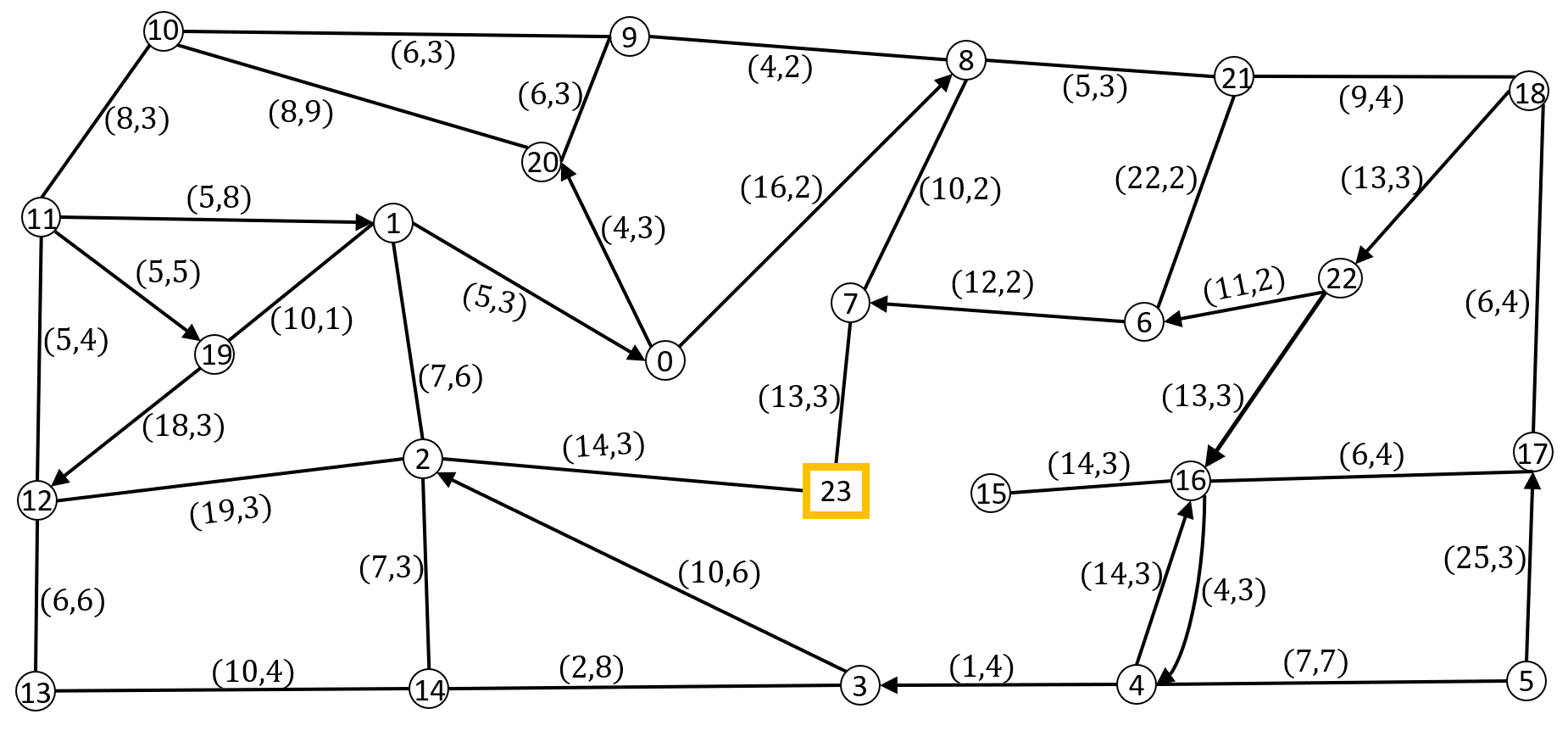}
	\caption{}\label{fig:23_node}
	\end{subfigure}
	\caption{Test graphs for simulations. (a) Graph I and (b)
          Graph II. Numbers in the circles represent node index, an
          arrow between nodes $l$ and $k$ indicates a directed arc
          from $l$ to $k$ and a line without an arrow between nodes
          $l$ and $k$ indicates a bi-directional edge. The tuple
          $(a,b)$ on the arc $(l,k)$ represents $(\rho_{lk}, t_{lk})$
          and the Hub node $H$ is marked in square.}\label{fig:graphs}
\end{figure}
%

\subsubsection{Variation of Route Set}

Here, we consider the feed-in problem for the two graphs in
Figure~\ref{fig:graphs}, that is $\Mbf{R} = \Rin$, and study the
variation in $\Mc{R}(\Rin)$ as a function of the time-window, $T$ and
the cost factor, $b$. Given $T$, we utilize a brute force method to
enumerate all the routes in a graph. Furthermore, with $b$ and $T$, we
utilize~\eqref{eq:alt-perceived-cost} and a V.o.T. of $\alpha = 0.5$
to generate the best alternate transportation time and cost, $\eta_l$,
$\zeta_l$ respectively, for each node $l$. Then, we
use~\eqref{eq:opt-pric} to generate prices for each service
tuple~$\stup$.

Figure~\ref{fig:norvsT} shows the variation of the number of routes in
$\Rin$ and $\MbrR$ as a function of $T$ for Graph~I and Graph~II. We
see an exponentially increasing trend for the number of routes in all
cases. However, as discussed in Remark~\ref{rem:reduced-routes}, one
can observe that the reduced route sets saturate eventually for large
enough $T$. Moreover, the value of $T$ at which this saturation occurs
increases with $b$. Finally, even prior to saturation, we see that
there is a significant reduction in the route set compared to the
original route set.
\begin{figure}[h]
	\begin{subfigure}[b]{0.25\textwidth}
		\centering \includegraphics[width=\textwidth]{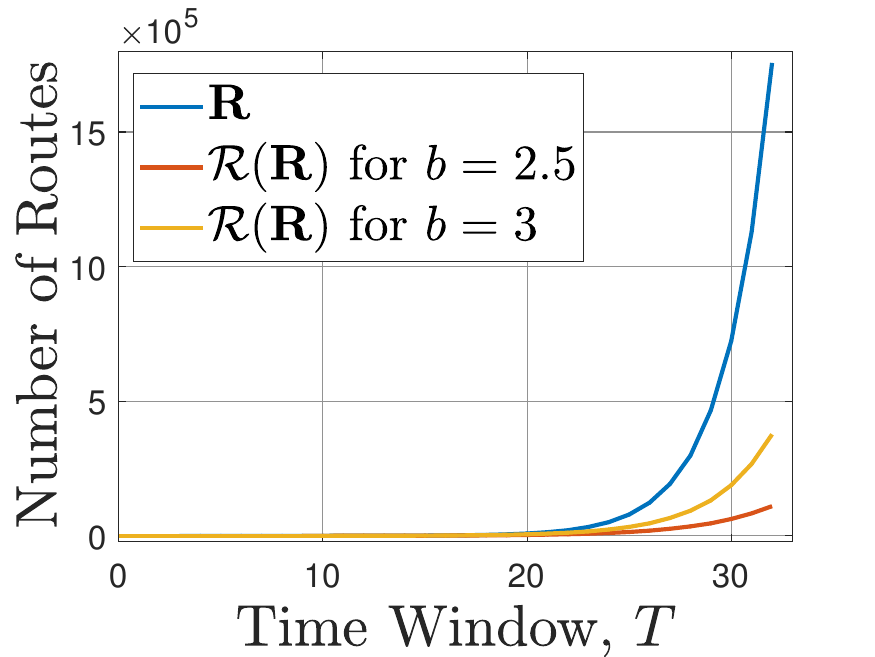}
		\caption{}
	\end{subfigure}%
	\begin{subfigure}[b]{0.25\textwidth}
		\centering \includegraphics[width=\textwidth]{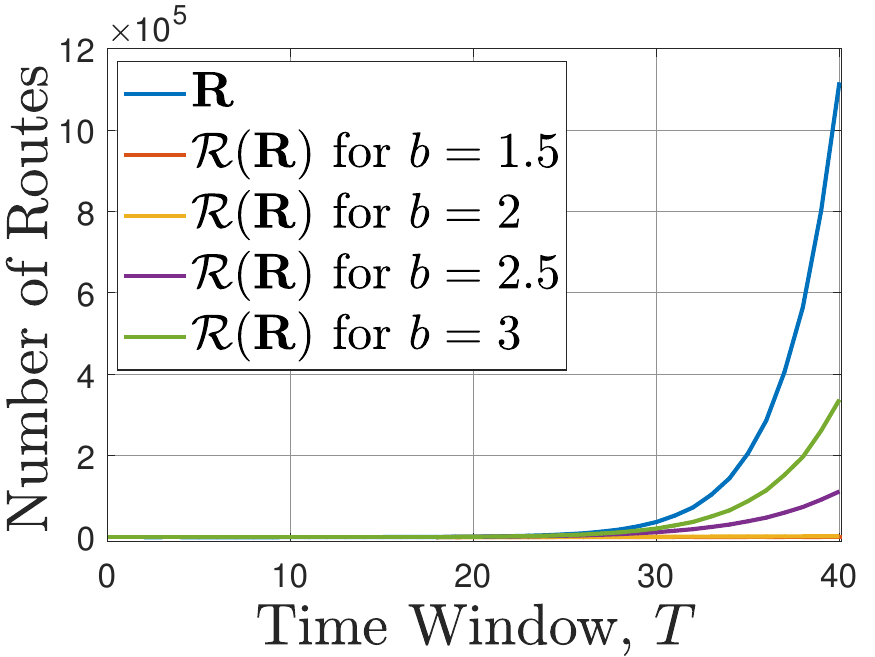}
		\caption{}
	\end{subfigure}
	\begin{subfigure}[b]{0.25\textwidth}
		\centering \includegraphics[width=\textwidth]{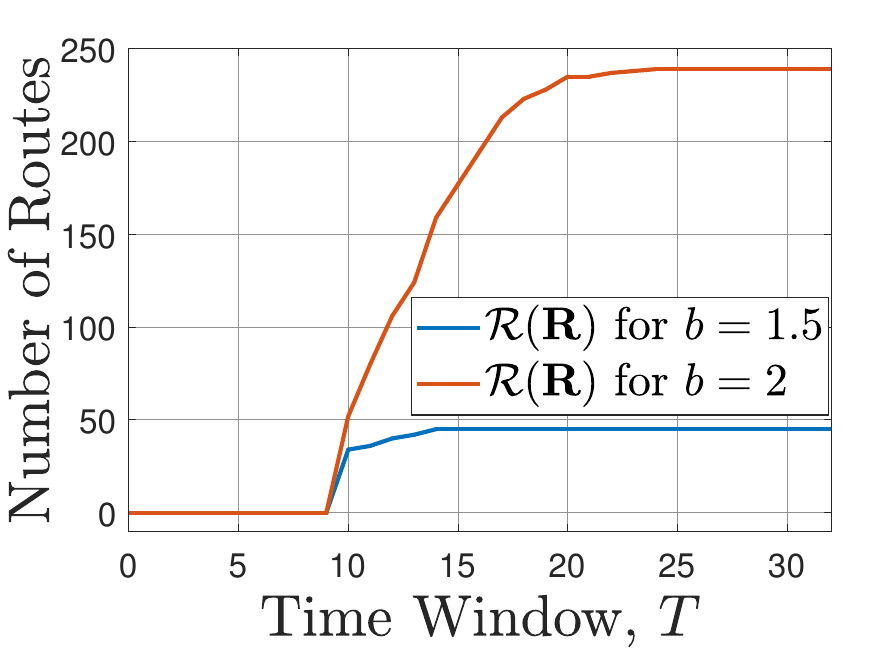}
		\caption{}
	\end{subfigure}%
	\begin{subfigure}[b]{0.25\textwidth}
		\centering \includegraphics[width=\textwidth]{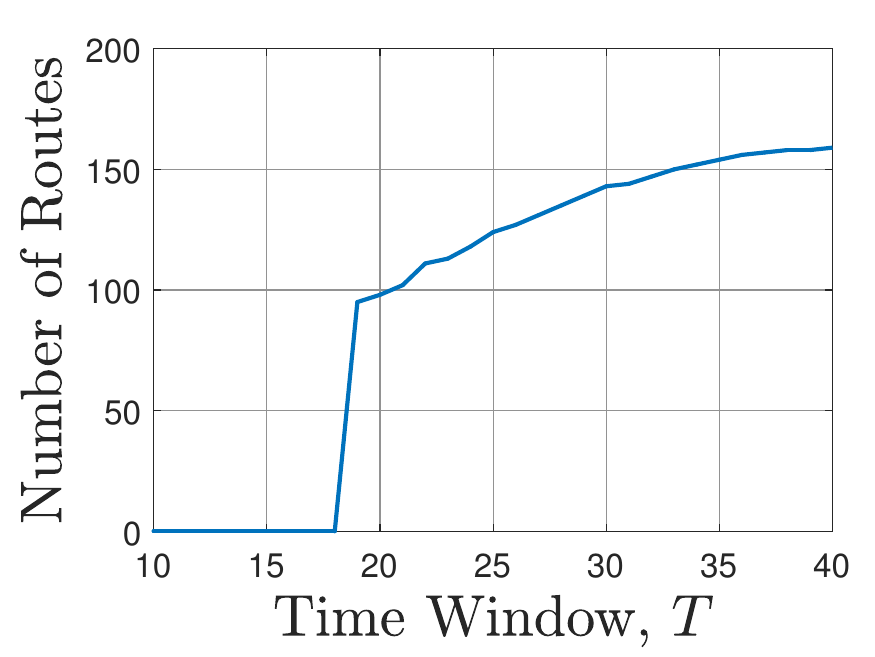}
		\caption{}
	\end{subfigure}
	\caption{ \textbf{(a), (b)}: Variation of number of
            routes in $\Mbf{R}=\Rin$ and $\MbrR$ as a function of $T$ for
          Graph~I and Graph~II, respectively.  \textbf{(c), (d)}:
          Saturation observed in $\MbrR$ with $b = 1.5, 2$ for Graph~I
          and with $b = 1.5$ for Graph~II.}\label{fig:norvsT}
\end{figure}

Figure~\ref{fig:bvsnor} shows the number of routes in $\MbrR$ as a
function of $b$ (with a step size of 0.01) for the graphs in
Figure~\ref{fig:graphs}.
\begin{figure}[h]
	\centering
	\begin{subfigure}[b]{0.5\textwidth}
		\centering \includegraphics[width=\textwidth]{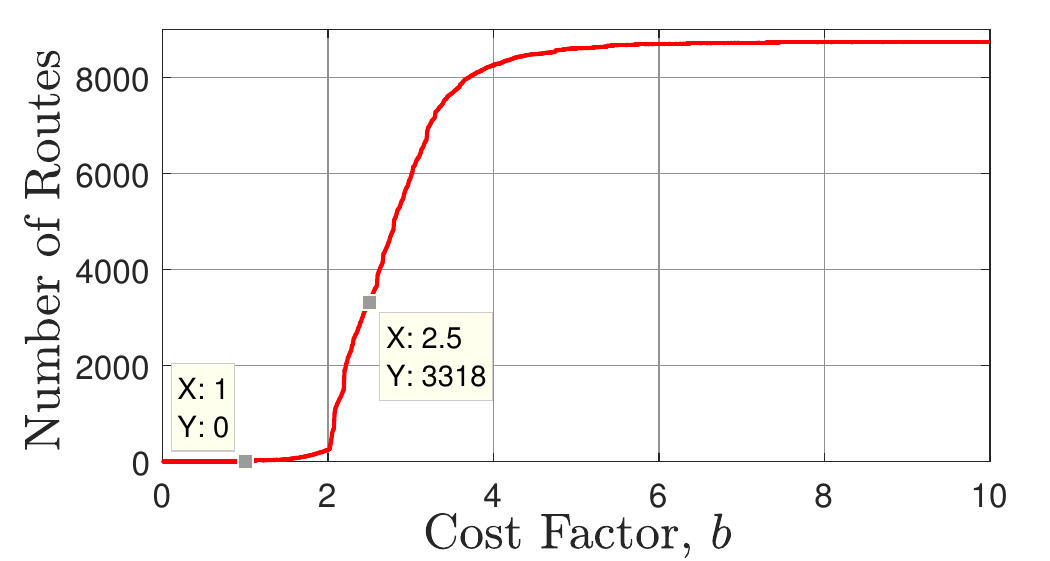}
		\caption{}
	\end{subfigure}%
	\begin{subfigure}[b]{0.5\textwidth}
		\centering \includegraphics[width=\textwidth]{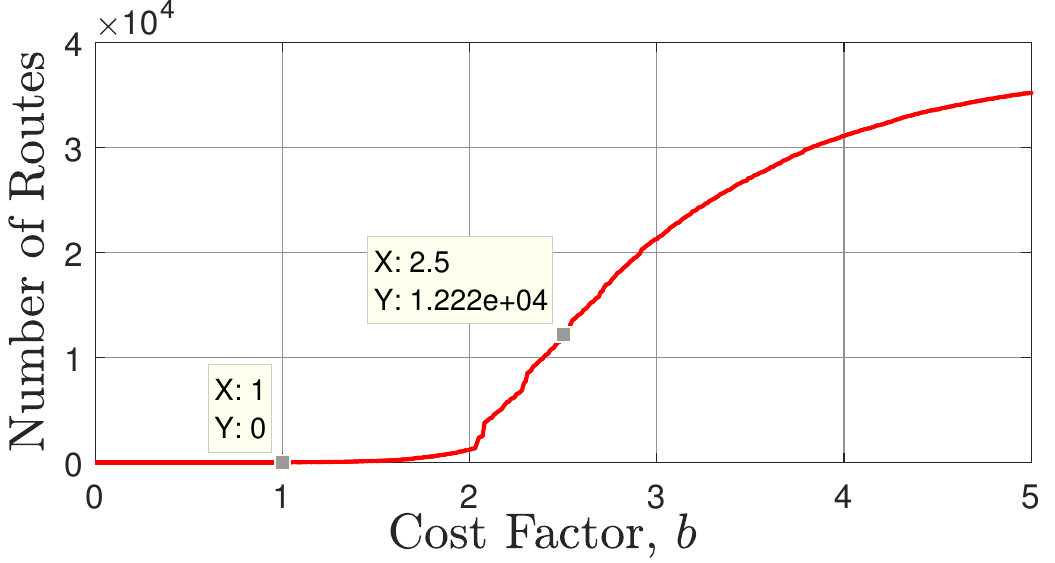}
		\caption{}
	\end{subfigure}
	\caption{Variation of number of routes in $\MbrR$ versus
          $b$ for $\alpha = 0.5$. Note that number of routes is 0 till
          $b=1$. \textbf{(a):} For Graph~I and $T= 20$. \textbf{(b):}
          For Graph~II and $T= 30$.}\label{fig:bvsnor}
\end{figure}
We note that the first route with origin as $H$ and the first
multi-legged route in $\MbrR$ occur at $b =
1.71$. Proposition~\ref{prop:necc-suff-b-multi-legs},~\ref{prop:necc-b-multi-legs}
and~\ref{prop:necc-b-multi-legs-kl>0} give necessary lower bounds on
$b$ for the existence of multi-legged routes in $\MbrR$ as $1.7027$,
$1.411$ and $1.692$, respectively. As one can see, $\hat{b}_l^*$ gives
a more conservative bound compared to $b_l^*$. In each case, the
origin of the multi-legged route is the node $l = 0$. In
Figure~\ref{fig:bvsnor}, we also see that there is a significant
increase in the number of routes around $b = 2.1$. This can be
explained by the fact that $b_l^* \in (2,2.1)$ for 5 nodes. 

\subsubsection{Feed-In and Comparison with Single-Depot Routing
  Problem for Graph II}

In the rest of this subsection, we present results for Graph II. We
set $b=2.5$ and $T = 30$, for which $\MbrR$ has 12219 routes and 45050
optimization variables. We also note that all nodes satisfy
Assumptions~\ref{A:2} and~\ref{A:4}. We simulated the \emph{feed-in}
problem for a fixed demand profile with $d_l$, for each $l \in V$,
drawn uniformly from $\intrangecc{0}{250}$. The total demand was
$\sum_l d_l = 2468$. Figure \ref{fig:Simulations} shows the simulation
results.  We utilized Proposition~\ref{prop:supply-opt} to generate
the route set $\Rin^-$ which had 6265 routes and the resulting number
of optimization variables was 12695. In Figure \ref{fig:sim_results1},
maximum profits for given total supply (marked in red line) is from
Proposition~\ref{prop:supply-opt}. The maximum profits as a function
of total supply $s$ converge to the absolute maximum profits,
$J_{\max} = 105033.5$ (given by Theorem \ref{th:suff-supplies}). We
also simulated an \emph{equivalent macroscopic} V.R.P., with all
supply at $H$, i.e.  $S_H = s$, $S_l = 0$, $\forall l\neq H$. We
observe in Figure~\ref{fig:sim_results1} that the profits earned in
this latter case are far lower, compared to that of any randomly
chosen supply configurations. This is explained by two factors -
insufficient time-window and cost-factor. Given $T=30$, 4 nodes do not
have $r\in \Rin$ such that $o_r = H$ and $l \in V_r$. Also given
$b = 2.5$, only 16 of the 23 nodes satisfy
Proposition~\eqref{prop:necc-suff-b-multi-legs}, implying at-least 7
nodes do not have $r \in \MbrR$ with $o_r = H$ and $l \in V_r$. The
necessary value of $b$ is $3.06$ for all nodes to satisfy
Proposition~\ref{prop:necc-suff-b-multi-legs}.
\begin{figure}[h]
	\centering
	\begin{subfigure}[b]{0.55\textwidth}
		\centering \includegraphics[
		width=\textwidth]{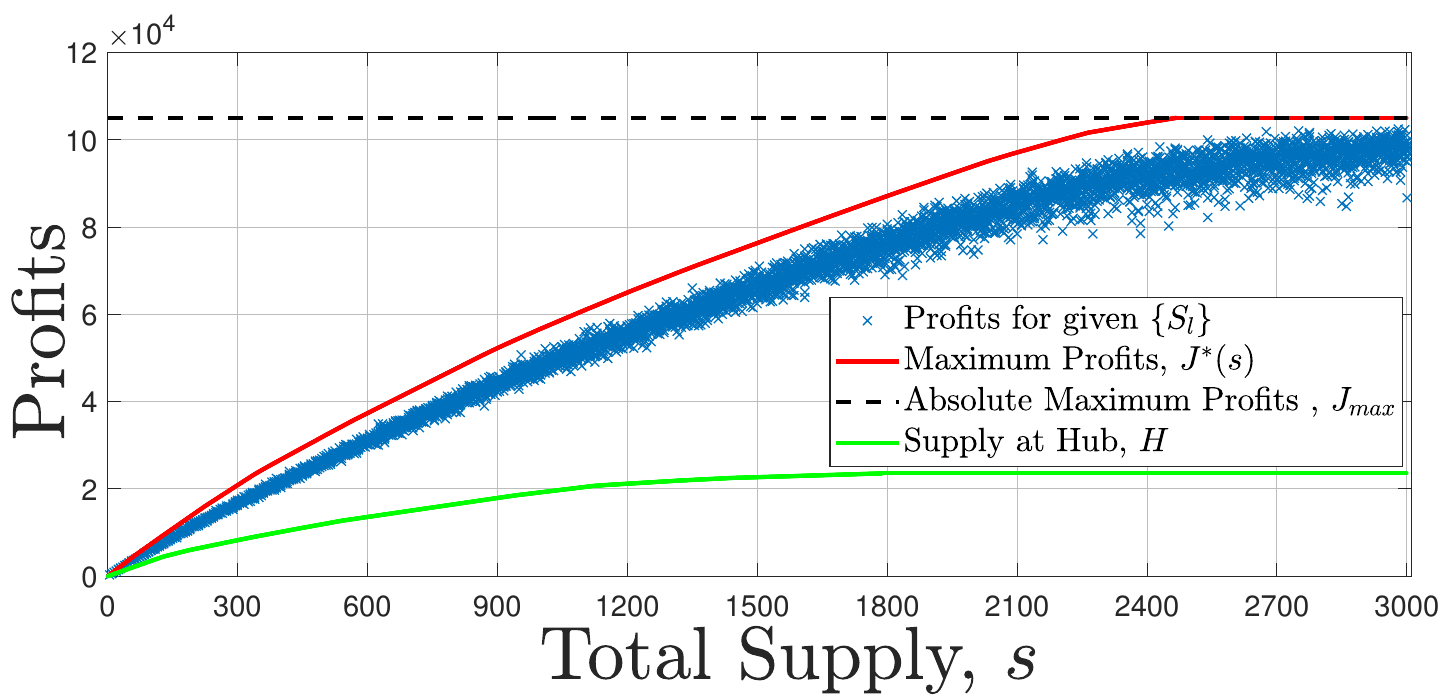}
		\caption{}
		\label{fig:sim_results1}
	\end{subfigure}%
	\begin{subfigure}[b]{0.45\textwidth}
		\centering \includegraphics[width=\textwidth]{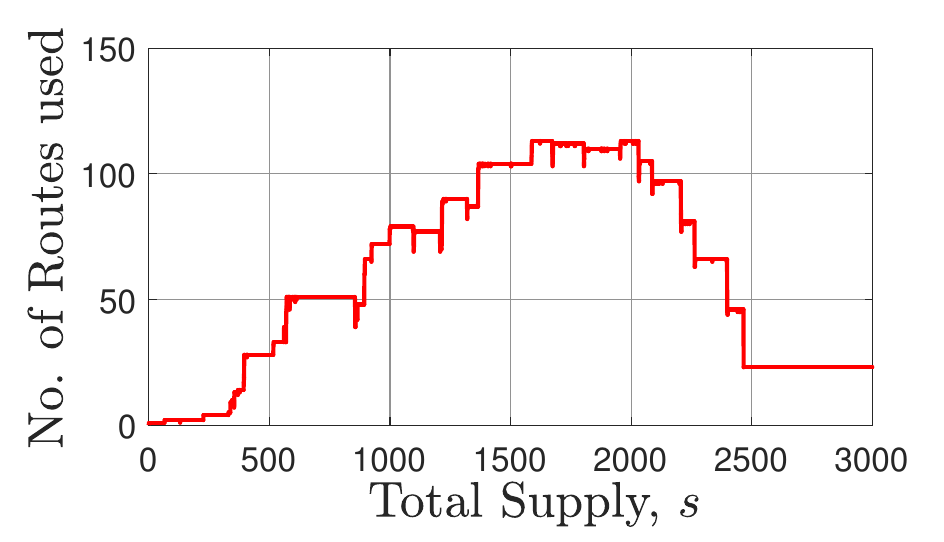}
		\caption{}
		\label{fig:sim_results3}
	\end{subfigure}%
	\caption{Simulation results for feed-in problem. \textbf{(a):}
          The red curve is the maximum profits as a function of total
          supply. The green curve is the maximum profits, with all
          supply at $H$, as a function of total supply $s$. The blue
          scatter points are for different supply configurations, with
          different supplies $s$. \textbf{(b):} Number of routes
          utilized for obtaining maximum profits in
          (a). } \label{fig:Simulations}
\end{figure}


In Figure \ref{fig:sim_results3}, we also see the number of routes
used to generate maximum profits generally increases with $s$ though
after a point the number of routes used starts to reduce. We imposed
the added restriction that the supply configuration is chosen from the
set \eqref{eq:supply-set} to compare with the equivalent feed-out
problem.

\subsubsection{Equivalence of Feed-Out and Feed-In Supply Optimization}

We use the directed graph in Figure~\ref{fig:23_node} and the
construction in Steps~\ref{C:supply-H}-\ref{C:reverse-demand} to
generate a feed-out problem from the feed-in supply optimization
problem. Using route set $\Rtl$, we generate the optimal profits for
the same instances of total supply as before and compared it with the
maximum profits for the feed-in supply optimization. The absolute
error is in the range of $10^{-4}$ while the maximum relative error is
$5.66\times10^{-6}$, which is within numerical tolerance given the
solver precision is $10^{-8}$ and the number of variables are
12695. This verifies the equivalence of the two problems.

\subsection{Microscopic Implementation}

In real life, the demands and supplies are indivisible entities,
specifically integers. Thus in this section, we implement the
microscopic optimization and compare it with the macroscopic
optimization. Specifically, in this subsection, we assume that for
each node $l$, $\tdl\in \natz$ and $\tSl \in \natz$ are the total
passengers and the number of vehicles, respectively. For simplicity,
we assume all the vehicles have the same capacity of $N$. In the
macroscopic problem, as we assumed that a unit demand can be serviced
by a unit of supply, in this subsection we assume that 1 unit of
demand and supply are $N$ passengers and 1 vehicle, respectively.
With this interpretation, we redefine problem~\eqref{eq:equivoptmodel}
with integer constraints. We let $\tfril$ be the number of people
serviced on $\stup$ and $\tfr$ be the number of vehicles taking route
$r$. We also assume that the operator can choose the supply
configuration as in Section~\ref{sec:max-profits}. To incorporate that
aspect, we assume that the operator has a total supply of $s$ at its
disposal. Then the constraints discussed in~\eqref{eq:constraints}
and~\eqref{eq:supply-opt} are modified to
\begin{equation}\label{eq:discrete-constrs}
  \begin{aligned}
    \tFl \ldef \sum_{(r,i)}\tfril &\leq \tdl, \quad \sum_{l \in
      V_r^i}\tfril \leq N\tfr, \quad \sum_{r\in \Mbf{R} \vert o_r = l
    }\tfr &\leq \tSl, \quad \sum_{l \in V} \tSl \leq s ,
  \end{aligned}
\end{equation}
along with the integer constraints on all optimization variables.

In~\eqref{eq:Nodeprofitability}, we assumed that $\bril$ is the
operator income per unit demand on the service $\stup$. Therefore, the
operator income per an individual passenger on the service $\stup$ is
$\bril/N$. Given this, the operator income from the service $\stup$ is
$\tfril \bril/N$. Further, the travelling cost incurred on route $r$
is $\tfr c_r$. Therefore, the \emph{microscopic operator profit
  maximization problem} is
\begin{align}
  \max_{\tfril, \tfr, \tSl} J &= \sum_{\stup \in \Mbf{W}(\Mbf{R})}
                                \tfril \bril/N -\sum_{r\in
                                \Mbf{R}}\tfr c_r \notag
  \\
  \text{s.t. } &\ \ \eqref{eq:discrete-constrs}, \quad  \tfril,
                 \tfr, \tSl, \tFl \in \natz
                 . \label{eq:discrete-opt}
\end{align} 

The corresponding macroscopic problem is the same
as~\eqref{eq:supply-opt}. We first compare the difference in the
optimal values of~\eqref{eq:discrete-opt} with $\Mbf{R} = \Rin$ and
$\Mbf{R} = \Mc{R}(\Rin)$ with simulations. We call these the
\emph{complete micro problem} and the \emph{reduced micro problem},
respectively. We first show that the reduced route set is also
effective for microscopic solutions. We simulated for five different
vehicle capacities: 2, 4, 5, 6 and 10. Note that with integer
constraints, the problem for the graphs in Figure~\ref{fig:graphs}
becomes rather large. Hence, for these set of simulations, we have
chosen a 12 node directed graph with a time window of $T = 25$ and
$b =2,\ 2.5\ \&\ 3$. For $\Mbf{R} = \Rin$, the number of routes was
$11694$ and the number of variables was 21127. On the other hand, the
problem with $\Mbf{R} = \Mc{R}(\Rin)$ has $75$ routes and $103$
variables for $ b= 2 $, $265$ routes and $510$ variables for
$ b = 2.5 $, and $1116$ routes and $2103$ variables for $b =
3$. Gurobi required $\approx 0.35$ to $0.5$ GB of memory for the
reduced micro problems (for all $b$) while it required $\approx 5$ GB
of memory for the complete micro problems. The 12-node graph data is given in 
the supplementary material. 

From Figure~\ref{fig:comp-times-error}, we observe that the
computation time ratios between the complete micro problem and the
reduced micro problem for different capacities is about $25$ to $ 40$
times with $b = 2$, $ 8$ to $15 $ times with $b = 2.5$ and $2$ to $3$
times for $b = 3$. Also, from Figure~\ref{fig:obj-val-error} we
observe that for most simulations, the optimal values returned by the
complete and reduced micro problems are very similar. Only in about
$4\%$ of the cases the value of the complete micro problem was higher
and even then only by about $10^{-4}$ times relative to the objective
value of the reduced micro problem. Since, Gurobi sometimes doesn't
converge and returns erroneous values due to lower values of MIP Gap
and numerical sensitivity~\cite{gurobi}, the MIP Gap was set to 0.001
to economize on time while ensuring reliability of results.
\begin{figure}
	\begin{subfigure}[t]{0.5\textwidth}
		\centering \includegraphics[width=\textwidth]{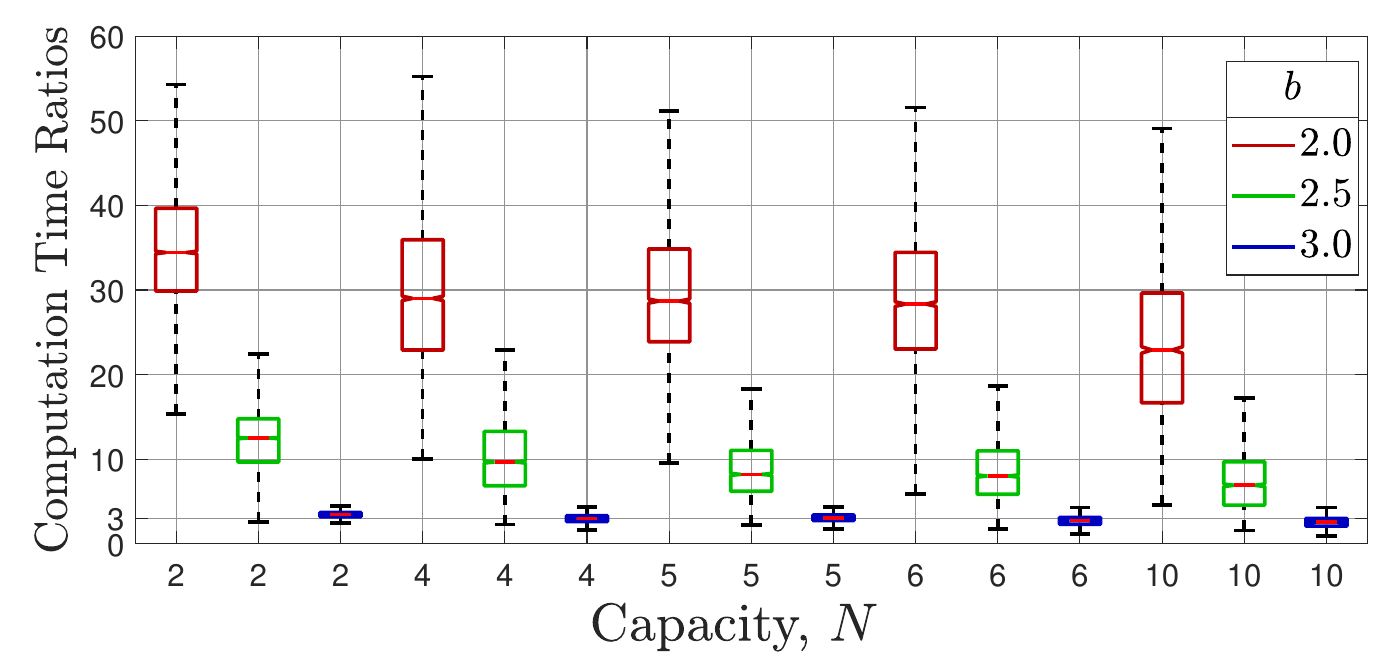}
		\caption{}
		\label{fig:comp-times-error}
	\end{subfigure}
	\begin{subfigure}[t]{0.5\textwidth}
		\centering \includegraphics[width =\textwidth]{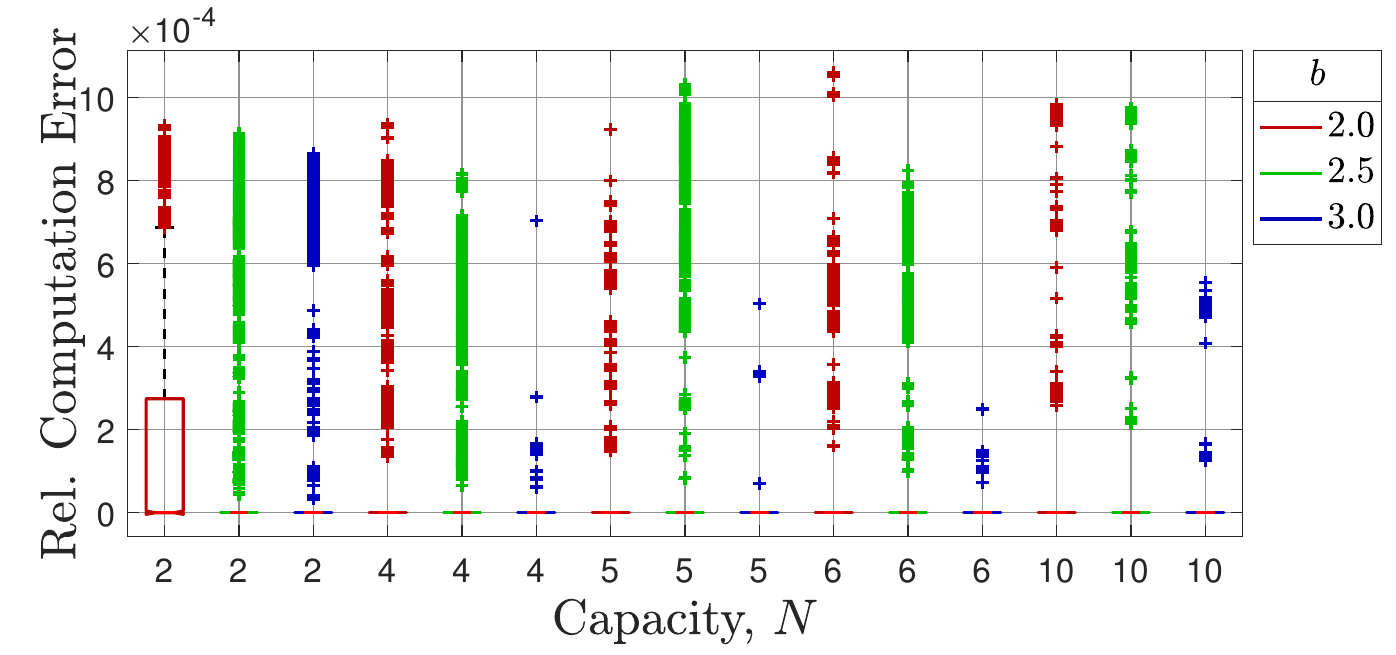}
		\caption{}
		\label{fig:obj-val-error}
	\end{subfigure}
	\caption{\textbf{(a)} Ratio of computation time for the
          complete micro problem to the computation time for the
          reduced micro problem for different $b$'s, and \textbf{(b)}
          Relative error of objective values between optimal solutions
          for the complete micro problem and the reduced micro problem
          on a 12 node graph for different values of $b$ and different
          capacities. The red, green and blue boxes represent
          $b = 2, 2.5, \text{ and } 3$ respectively. }
	\label{fig:Reduced-Micro-vs-micro}
\end{figure}

Next, we compare the macroscopic problem and the complete microscopic
problem for the graphs in Figure~\ref{fig:graphs}. We again set time
windows of $T = 20$ for Graph~I and $T = 30$ for Graph~II and
$b = 2.5$ and $\alpha = 0.5$ for both the graphs. We chose 10
different demand configurations randomly for both the graphs in
Figure~\ref{fig:graphs}.

Since the complete microscopic problem is very demanding in terms of
both the computation time and memory requirements, we do the
comparison in an indirect manner. Notice that the macroscopic problem
is a relaxation of both the complete and reduced micro
problems. Hence, the solutions to the macroscopic problem and the
reduced micro problem provide an upper and lower bound, respectively,
on the optimal objective value for the complete micro problem. Thus,
by comparing these bounds, we can also indirectly compare the
solutions to the macroscopic problem and the complete micro problem.

In general it is expected that with higher values of capacity $N$, the
optimal values of the solutions will be far-off from the solutions
obtained from~\eqref{eq:equivoptmodel} and the supply optimization
problem proposed in Section~\ref{sec:max-profits}. Thus, for a
successful implementation, the relaxation needs to find optimal
solutions which are ``\emph{close}" to the integer optimal.  
 
We first note the actual range of computation times of the reduced
micro problem in Figure~\ref{fig:comp-times}. We observe that the
computation times generally tend to increase with $N$. This is because
more branch and bound iterations are necessary to arrive at the
integer optimal solution.
Next, in Figure~\ref{fig:Rel-errors}, we notice that the relative
difference between the optimal values of reduced micro problem and the
macroscopic problem is quite small, with a median of the order of
$10^{-3}$. This is a reasonably small value considering the number of
variables, the order of the optimal objective value and the
computation power utilized for computing these solutions. Recall that
the solutions to the macroscopic problem and the reduced micro problem
provide upper and lower bounds, respectively, on the optimal objective
value in the complete micro problem. Therefore, the relative
difference between the complete micro problem and the macroscopic
problem is smaller than the errors in
Figure~\ref{fig:Rel-errors}. Thus, the macroscopic problem is a
reasonably good approximation of the reduced micro problem and thereby
the complete micro problem as well.  We also observed that when the
demand is an integer multiple of the supply, there exists an integer
solution. Furthermore, the solution was calculated in times similar to
the macroscopic problem. We suspect that this behaviour is because of
the fact that a unit of supply can pick-up a unit demand at one
location itself and thereby no inter-node mode-sharing is required.

\begin{figure}[htb]
	\centering
	\begin{subfigure}[t]{0.35\textwidth}
		\centering \includegraphics[width =\textwidth]{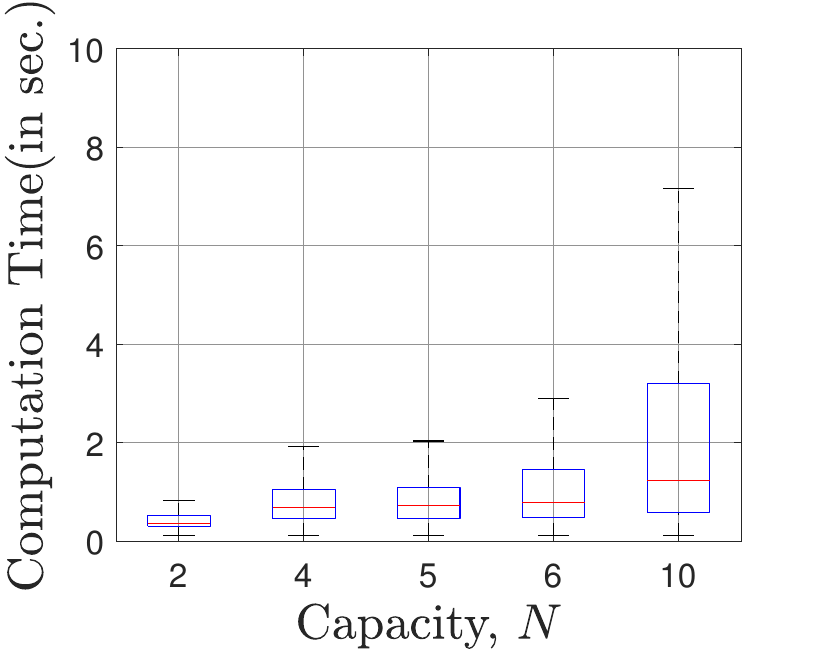}
		\caption{}
		\label{fig:comp-times-grph-1}
	\end{subfigure}%
	\begin{subfigure}[t]{0.35\textwidth}
		\centering \includegraphics[width =\textwidth]{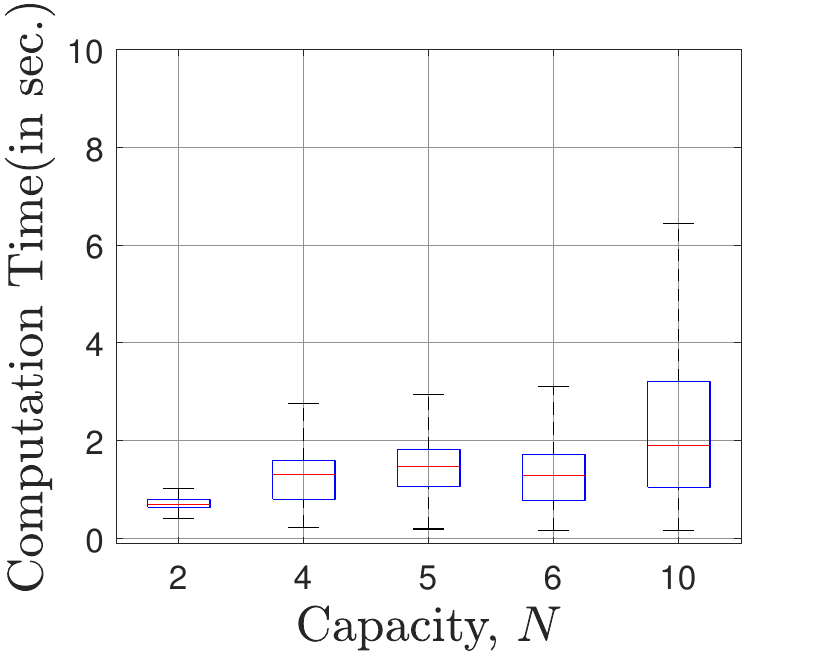}
		\caption{}
		\label{fig:comp-times-grph-2}
	\end{subfigure}
	\caption{Computation times of optimal solutions of reduced
          micro problem for Graph~I and~II in (a) and (b),
          respectively. } \label{fig:comp-times}
\end{figure} 
\begin{figure}[htb]
	\centering
	\begin{subfigure}[t]{0.35\textwidth}
		\centering \includegraphics[width =\textwidth]{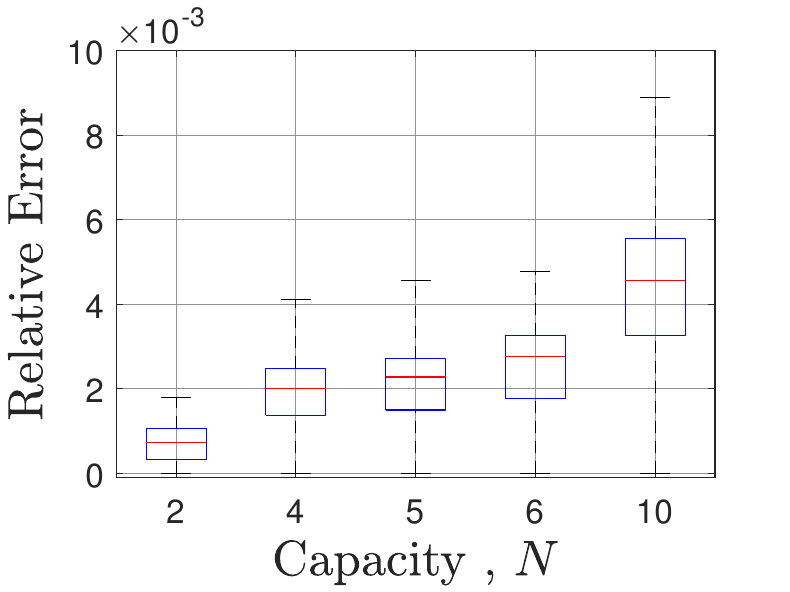}
		\caption{}
		\label{fig:Rel-error-grph-1}
	\end{subfigure}%
	\begin{subfigure}[t]{0.35\textwidth}
		\centering \includegraphics[width =\textwidth ]{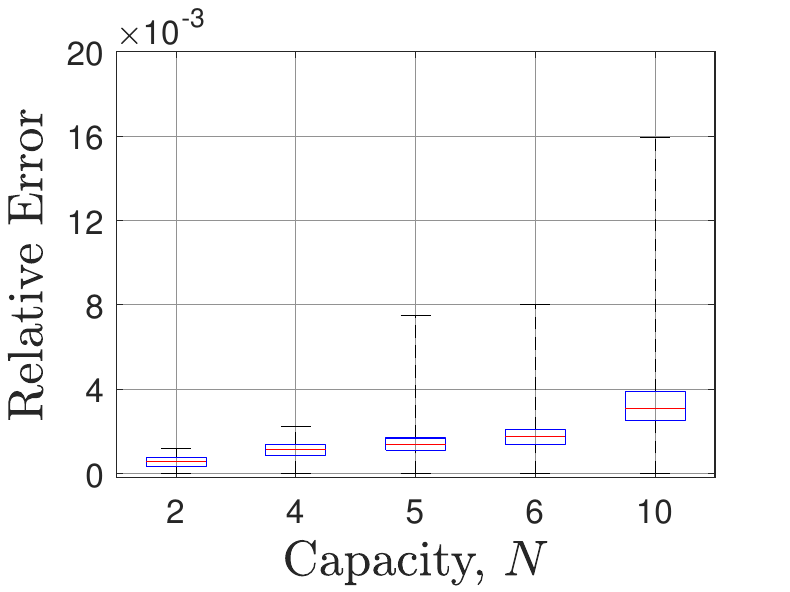}
		\caption{}
		\label{fig:Rel-error-grph-2}
	\end{subfigure}
        \caption{Relative errors between the optimal objective values
          of macroscopic and reduced micro problems for Graphs~I
          and~II in (a) and (b), respectively.}\label{fig:Rel-errors}
\end{figure}

\subsection{Pre-positioning of Relief Material for Disaster %
  Response}

As we discussed in Remark~\ref{rem:op-cost} and in
Section~\ref{sec:pricing-scenarios}, our model can be used even in the
context of routing and planning for disaster response, where the
objective is not to maximize monetary profits but to maximize social
welfare. In this subsection, we consider the problem of planning the
routing of relief supplies post a natural disaster, such as a cyclone.
The states on the eastern coast of India - West Bengal, Odisha and
Andhra Pradesh - face devastating cyclones quite frequently. The scale
of the disaster response is massive and often millions of people are
evacuated~\cite{BBC-Fani-2019, BBC-Amphan-2020,
  WP-Amphan-2020}. Post-cyclone relief efforts are also on an equally
massive scale.
The most vulnerable areas generally tend to be low-lying and near the
coast.

For the illustrative example here, we chose 12 coastal districts in
the state of Odisha that have been severely affected in recent
cyclones. The setup is as follows. The nodes in the graph are the
districts and the links between the nodes are the major highways
between the district headquarters. We chose the hub node as
Bhubaneshwar, the capital city of the state of Odisha. We chose the
link costs and travel times proportional to the distance, which we
created using the \emph{Google My Maps} tool. We present the map and
the constructed network in Figure~\ref{fig:or-map}. Though the links
are marked bi-directional, there are some minor asymmetries. We
present the link data in the supplementary material. We look at the
problem of transporting of relief material from the hub node to the
remaining nodes.

We chose the demand configuration as follows. We assume that the
number of affected people is 2\% of the total population of each
district. Furthermore, we assume that a unit demand refers to relief
material for 400 people. We round off the demands to the nearest
multiple of 5 and we added 10 extra units as a buffer. We assume that
loading of trucks takes 2 hours. We assume that initially, all the
supply of trucks is stationed at the hub node. Thus, for multi-legged
routes we add a wait time of 2 hours per leg. The main goal here is to
maximize the demand met while reducing transportation costs.

With this setup, we present two sets of simulation results. In the
first set of simulations, for each node $l$ we chose the value of
$k_l$ randomly in $\intrangecc{-5}{-1}$. With these settings, we
varied the total supply $s$ and the length of the time window $T$. We
show the results on the demand met in
Figure~\ref{fig:or-dem-met-T}. As can be seen, the obvious trend is
that for a fixed $T$ more demand can be met given greater
supply. Similarly, for each fixed total supply $s$, longer
time-windows mean more demand is met.
 \begin{figure}
 	\centering
 	\includegraphics[width=0.35\textwidth]{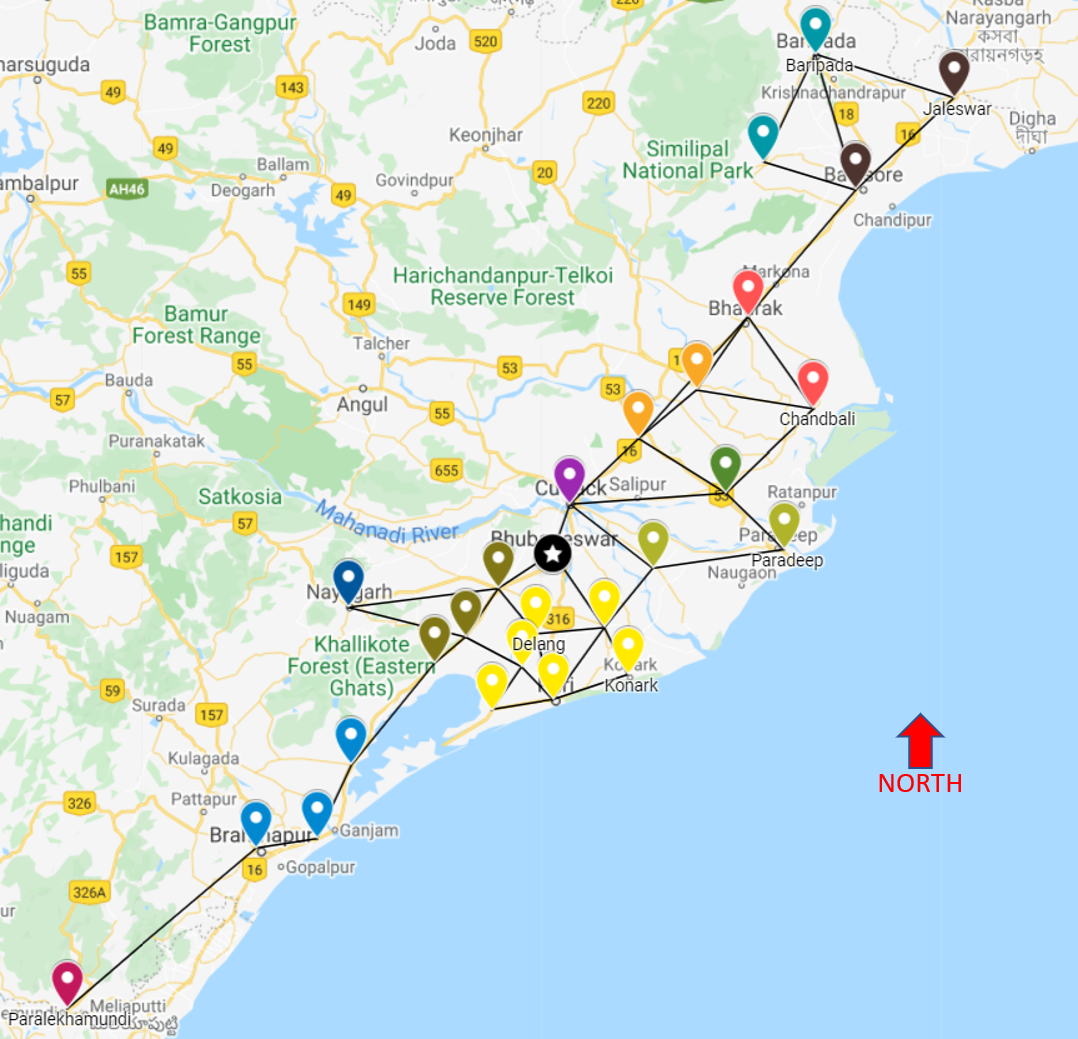}
 	\caption{Map of the graph utilized for Odisha. This map was created with \emph{Google My Maps} tool. Markers of same colour represent places in same district. The Hub (Bhubaneswar) is marked with a black star. }\label{fig:or-map}
 \end{figure}

 In the second set of simulations, we fix the window time $T = 12$ and
 fix the values of $k_l$ to the same value for all nodes
 $l$. Figure~\ref{fig:or-dem-met-kl} shows the simulation results for
 different values of supply $s$ and different values of $k_l$. For
 each fixed supply, as we increase the magnitude of $k_l$ more demand
 is met. Thus, in general, $k_l$ can be made proportional to the
 severity of the cyclone in a given region/node $l$ to divert the
 supplies accordingly.
\begin{figure}[h]
	\begin{subfigure}[b]{0.5\textwidth}
		\centering \includegraphics[width =\textwidth]{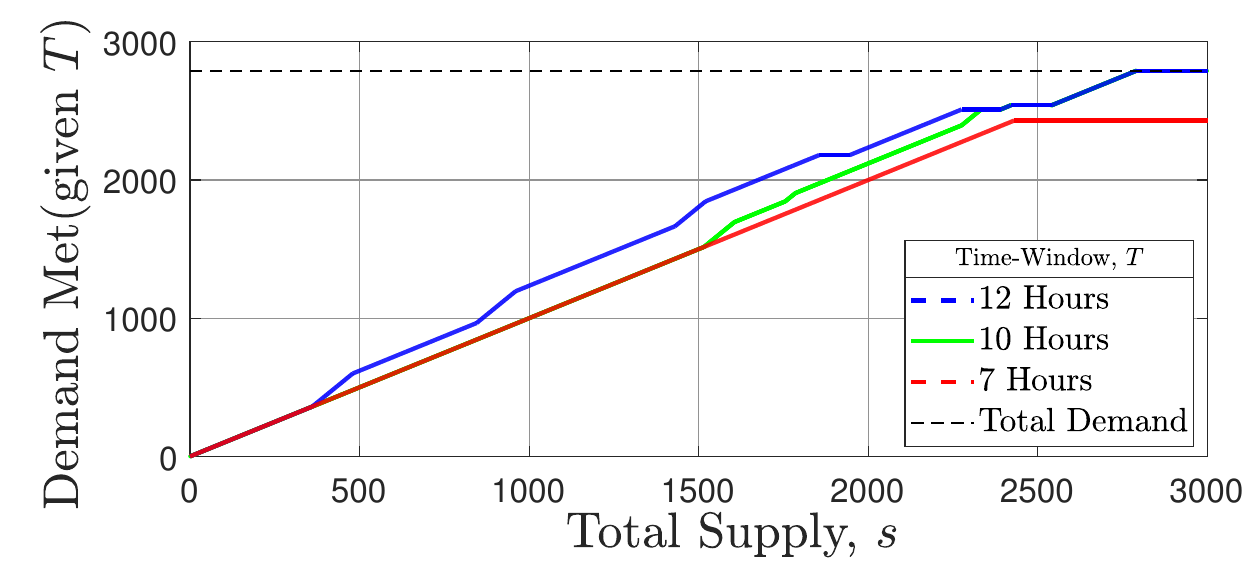}
		\caption{}
		\label{fig:or-dem-met-T}
	\end{subfigure}%
	\begin{subfigure}[b]{0.5\textwidth}
		\centering \includegraphics[width=\textwidth]{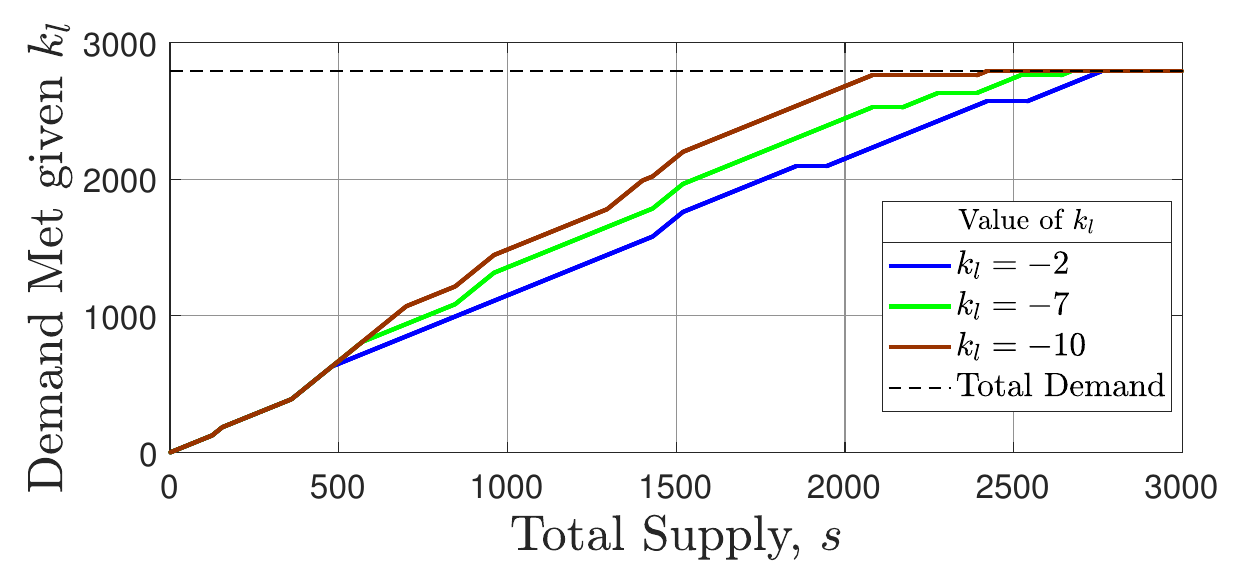}
		\caption{}
		\label{fig:or-dem-met-kl}
	\end{subfigure}%
	\caption{Simulation results for the Odisha supply prepositioning model. \textbf{(a):}
		Demand met given different window times $T$.
		\textbf{(b):} Demand met given window time, $T = 12 $ hours and different values of $k_l$. } \label{fig:Or-SOFO-RES}
\end{figure}


\section{Conclusions}\label{sec:conc}

In this paper, we proposed a problem of \emph{one-shot} coordination
of first mode or last mode transportation service, wherein an operator
seeks to maximize its profits or social welfare with routing and
allocation, for transporting a known demand to or from a common
destination on a network in a given fixed time window. We posed the
problem in a macroscopic setting where we considered all supplies and
demands as volumes. Using K.K.T. analysis we were able to design an
offline (supply and demand independent) method that reduces the
complexity of the online (after supply and demand are revealed)
optimization. Then, we considered the feed-in supply optimization
problem, analyzed its properties and computed the absolute maximum
profits that the operator can earn over all possible supply
configurations for a given demand configuration. We showed an
equivalence between the feed-in supply optimization problem and the
\emph{one-shot feed-out} problem, wherein the operator needs to
drop-off people to their destinations from a common origin within a
fixed time window. This allows us to directly apply the results and
algorithms developed for the feed-in problem. We presented a pricing
model and derived necessary conditions for the feeder service to be
viable. Finally, we presented several simulation results to illustrate
our analytical results. Through simulations, we also demonstrated that
the route reduction algorithm that we proposed for the macroscopic
formulation is still useful for efficiently computing the solutions to
the microscopic formulation, in which the decision variables are all
integers. We also explored a realistic scenario of disaster response.

Future work includes improvements in the route reduction algorithm,
perhaps making it also dependent on the demand and supply
configurations, a more rigorous study of the macroscopic formulation
and the route reduction algorithm as a computationally efficient
solution method for the microscopic formulation; extension to a
multiple time window problem, load balancing of the supply in
accordance with the anticipated demand using the insights from the
feed-in supply optimization problem, extension to the scenario with
uncertainty about supply and demand and finally an integrated
coordination of multiple modes of transportation.

\section{Acknowledgements}
We wish to thank Dr. Tarun Rambha of Department of Civil Engineering
at the Indian Institute of Science (IISc), Bengaluru for his valuable
comments and suggestions.

\bibliographystyle{IEEEtran} \bibliography{alias,main,pavan}

\begin{thebibliography}{10}
\providecommand{\url}[1]{#1}
\csname url@samestyle\endcsname
\providecommand{\newblock}{\relax}
\providecommand{\bibinfo}[2]{#2}
\providecommand{\BIBentrySTDinterwordspacing}{\spaceskip=0pt\relax}
\providecommand{\BIBentryALTinterwordstretchfactor}{4}
\providecommand{\BIBentryALTinterwordspacing}{\spaceskip=\fontdimen2\font plus
\BIBentryALTinterwordstretchfactor\fontdimen3\font minus
  \fontdimen4\font\relax}
\providecommand{\BIBforeignlanguage}[2]{{%
\expandafter\ifx\csname l@#1\endcsname\relax
\typeout{** WARNING: IEEEtran.bst: No hyphenation pattern has been}%
\typeout{** loaded for the language `#1'. Using the pattern for}%
\typeout{** the default language instead.}%
\else
\language=\csname l@#1\endcsname
\fi
#2}}
\providecommand{\BIBdecl}{\relax}
\BIBdecl

\bibitem{SG-PT:2019-ecc}
S.~Goswami and P.~Tallapragada, ``One-shot coordination of feeder vehicles for
  multi-modal transportation,'' in \emph{{E}uropean {C}ontrol {C}onference},
  Naples, Italy, 2019, pp. 1714--1719.

\bibitem{VC-RB-AB:2012}
V.~Campos, R.~Bandeira, and A.~Bandeira, ``A method for evacuation route
  planning in disaster situations,'' \emph{Procedia-Social and Behavioral
  Sciences}, vol.~54, pp. 503--512, 2012.

\bibitem{RL-RM:2000}
R.~Lave and R.~R.~Mathias, ``State of the art of paratransit,''
  \emph{Transportation in the New Millennium}, vol. 478, pp. 1--7, 2000.

\bibitem{MSS-etal:2014}
M.~SteadieSeifi, N.~Dellaert, W.~Nuijten, T.~V. Woensel, and R.~Raoufi,
  ``Multimodal freight transportation planning: A literature review,''
  \emph{European journal of operational research}, vol. 233, no.~1, pp. 1--15,
  2014.

\bibitem{CB-etal:2002}
C.~Barnhart, N.~Krishnan, D.~Kim, and K.~Ware, ``Network design for express
  shipment delivery,'' \emph{Computational Optimization and Applications},
  vol.~21, no.~3, pp. 239--262, 2002.

\bibitem{HM:1989}
H.~Min, ``The multiple vehicle routing problem with simultaneous delivery and
  pick-up points,'' \emph{Transportation Research Part A: General}, vol.~23,
  no.~5, pp. 377--386, 1989.

\bibitem{PT-DV:2002-book}
P.~Toth and D.~Vigo, \emph{The vehicle routing problem}.\hskip 1em plus 0.5em
  minus 0.4em\relax SIAM, 2002.

\bibitem{FF:2013-book}
F.~Ferrucci, \emph{Pro-active Dynamic Vehicle Routing: Real-time Control and
  Request-forecasting Approaches to Improve Customer Service}.\hskip 1em plus
  0.5em minus 0.4em\relax Springer Science \& Business Media, 2013.

\bibitem{MWU:2017-book}
M.~W. Ulmer, \emph{Approximate Dynamic Programming for Dynamic Vehicle
  Routing}, 1st~ed.\hskip 1em plus 0.5em minus 0.4em\relax Springer
  International Publishing, 2017.

\bibitem{TV-GL-PM-2019-VRP-Review}
T.~Vidal, G.~Laporte, and P.~Matl, ``A concise guide to existing and emerging
  vehicle routing problem variants,'' \emph{European Journal of Operational
  Research}, 2019.

\bibitem{NA-AE-MS-XW:2012}
N.~Agatz, A.~Erera, M.~Savelsbergh, and X.~Wang, ``Optimization for dynamic
  ride-sharing: A review,'' \emph{European Journal of Operational Research},
  vol. 223, no.~2, pp. 295--303, 2012.

\bibitem{JA-etal:2017}
J.~Alonso-Mora, S.~Samaranayake, A.~Wallar, E.~Frazzoli, and D.~Rus,
  ``On-demand high-capacity ride-sharing via dynamic trip-vehicle assignment,''
  \emph{Proceedings of the National Academy of Sciences}, vol. 114, no.~3, pp.
  462--467, 2017.

\bibitem{FYV-etal:2018}
F.~Vincent, S.~Purwanti, A.~Redi, C.~Lu, S.~Suprayogi, and P.~Jewpanya,
  ``Simulated annealing heuristic for the general share-a-ride problem,''
  \emph{Engineering Optimization}, vol.~50, no.~7, pp. 1178--1197, 2018.

\bibitem{CCT-CYC:2007}
C.~Tao and C.~Chen, ``Heuristic algorithms for the dynamic taxipooling problem
  based on intelligent transportation system technologies,'' in \emph{4th
  International conference on fuzzy systems and knowledge}.\hskip 1em plus
  0.5em minus 0.4em\relax IEEE, 2007, pp. 590--595.

\bibitem{JFC-GL:2007}
J.~Cordeau and G.~Laporte, ``The dial-a-ride problem: models and algorithms,''
  \emph{Annals of operations research}, vol. 153, no.~1, pp. 29--46, 2007.

\bibitem{SCH-etal:2018}
S.~Ho, W.~Szeto, Y.~Kuo, J.~Leung, M.~Petering, and T.~Tou, ``A survey of
  dial-a-ride problems: Literature review and recent developments,''
  \emph{Transportation Research Part B: Methodological}, 2018.

\bibitem{BL-DK-TVW-HAR:2016}
B.~Li, D.~Krushinsky, H.~Reijers, and T.~V. Woensel, ``The share-a-ride problem
  with stochastic travel times and stochastic delivery locations,''
  \emph{Transportation Research Part C: Emerging Technologies}, vol.~67, pp.
  95--108, 2016.

\bibitem{KTS-NHD-DHL:2010}
K.~Seow, N.~Dang, and D.~Lee, ``A collaborative multiagent taxi-dispatch
  system,'' \emph{IEEE Transactions on Automation Science and Engineering},
  vol.~7, no.~3, pp. 607--616, 2010.

\bibitem{FM-etal:2016}
F.~Miao, S.~Han, S.~Lin, J.~Stankovic, D.~Zhang, S.~Munir, H.~Huang, T.~He, and
  G.~Pappas, ``Taxi dispatch with real-time sensing data in metropolitan areas:
  A receding horizon control approach,'' \emph{IEEE Transactions on Automation
  Sciences and Engineering}, vol.~13, no.~2, pp. 463--478, 2016.

\bibitem{FR:2017}
F.~Rossi, R.~Zhang, Y.~Hindy, and M.~Pavone, ``Routing autonomous vehicles in
  congested transportation networks: structural properties and coordination
  algorithms,'' \emph{Autonomous Robots}, pp. 1--16, 2017.

\bibitem{MS-etal:2018}
M.~Salazar, F.~Rossi, M.~Schiffer, C.~Onder, and M.~Pavone, ``On the
  interaction between autonomous mobility-on-demand and public transportation
  systems,'' \emph{arXiv preprint arXiv:1804.11278}, 2018.

\bibitem{MP-SLS-EF-DR:2012}
M.~Pavone, S.~Smith, E.~Frazzoli, and D.~Rus, ``Robotic load balancing for
  mobility-on-demand systems,'' \emph{International Journal of Robotics
  Research}, vol.~31, no.~7, pp. 839--854, 2012.

\bibitem{RZ-MP:2016}
R.~Zhang and M.~Pavone, ``Control of robotic mobility-on-demand systems: a
  queueing-theoretical perspective,'' \emph{The International Journal of
  Robotics Research}, vol.~35, no. 1-3, pp. 186--203, 2016.

\bibitem{AV-WG-SS:2017}
A.~Vakayil, W.~Gruel, and S.~Samaranayake, ``Integrating shared-vehicle
  mobility-on-demand systems with public transit,'' in \emph{96th Annual
  Meeting of Transportation Research Board}.\hskip 1em plus 0.5em minus
  0.4em\relax T.R.B., 2017.

\bibitem{SS-NC:2016}
S.~Shaheen and N.~Chan, ``Mobility and the sharing economy: Potential to
  facilitate the first-and last-mile public transit connections,'' \emph{Built
  Environment}, vol.~42, no.~4, pp. 573--588, 2016.

\bibitem{AMC-etal-2012-Review}
A.~M. Caunhye, X.~Nie, and S.~Pokharel, ``Optimization models in emergency
  logistics: A literature review,'' \emph{Socio-economic planning sciences},
  vol.~46, no.~1, pp. 4--13, 2012.

\bibitem{GGP-RB-2013-Review}
G.~Galindo and R.~Batta, ``Review of recent developments in or/ms research in
  disaster operations management,'' \emph{European Journal of Operational
  Research}, vol. 230, no.~2, pp. 201--211, 2013.

\bibitem{MSH-etal-2016}
M.~S. Habib, Y.~H. Lee, and M.~S. Memon, ``Mathematical models in humanitarian
  supply chain management: A systematic literature review,'' \emph{Mathematical
  Problems in Engineering}, vol. 2016, 2016.

\bibitem{MS-RB-QH-2020-Review}
M.~Sabbaghtorkan, R.~Batta, and Q.~He, ``Prepositioning of assets and supplies
  in disaster operations management: Review and research gap identification,''
  \emph{European Journal of Operational Research}, vol. 284, no.~1, pp. 1--19,
  2020.

\bibitem{MA-etal-2015}
M.~Ahmadi, A.~Seifi, and B.~Tootooni, ``A humanitarian logistics model for
  disaster relief operation considering network failure and standard relief
  time: A case study on san francisco district,'' \emph{Transportation Research
  Part E: Logistics and Transportation Review}, vol.~75, pp. 145--163, 2015.

\bibitem{MH-etal-2012}
M.~Huang, K.~Smilowitz, and B.~Balcik, ``Models for relief routing: Equity,
  efficiency and efficacy,'' \emph{Transportation research part E: logistics
  and transportation review}, vol.~48, no.~1, pp. 2--18, 2012.

\bibitem{JMF-etal-2018}
J.~M. Ferrer, F.~J. Mart{\'\i}n-Campo, M.~T. Ortu{\~n}o, A.~J.
  Pedraza-Mart{\'\i}nez, G.~Tirado, and B.~Vitoriano, ``Multi-criteria
  optimization for last mile distribution of disaster relief aid: Test cases
  and applications,'' \emph{European Journal of Operational Research}, vol.
  269, no.~2, pp. 501--515, 2018.

\bibitem{AMC-etal-2008}
A.~M. Campbell, D.~Vandenbussche, and W.~Hermann, ``Routing for relief
  efforts,'' \emph{Transportation Science}, vol.~42, no.~2, pp. 127--145, 2008.

\bibitem{RL-ZJ-2019}
R.~Liu and Z.~Jiang, ``A hybrid large-neighborhood search algorithm for the
  cumulative capacitated vehicle routing problem with time-window
  constraints,'' \emph{Applied Soft Computing}, vol.~80, pp. 18--30, 2019.

\bibitem{SD-SB:2016}
S.~Diamond and S.~Boyd, ``{CVXPY}: A {P}ython-embedded modeling language for
  convex optimization,'' \emph{Journal of Machine Learning Research}, vol.~17,
  no.~83, pp. 1--5, 2016.

\bibitem{gurobi}
\BIBentryALTinterwordspacing
L.~Gurobi~Optimization, ``Gurobi optimizer reference manual,'' 2021. [Online].
  Available: \url{http://www.gurobi.com}
\BIBentrySTDinterwordspacing

\bibitem{BBC-Fani-2019}
\BIBentryALTinterwordspacing
B.~S. Correspondent, ``India cyclone fani evacuation efforts hailed a
  success,'' \emph{BBC}, 4 May, 2019. [Online]. Available:
  \url{https://www.bbc.com/news/world-asia-48160096}
\BIBentrySTDinterwordspacing

\bibitem{BBC-Amphan-2020}
\BIBentryALTinterwordspacing
L.~Lear, ``Amphan: India and bangladesh evacuate millions ahead of super
  cyclone,'' \emph{BBC}, 19 May, 2020. [Online]. Available:
  \url{https://www.bbc.com/news/world-asia-india-52718826}
\BIBentrySTDinterwordspacing

\bibitem{WP-Amphan-2020}
\BIBentryALTinterwordspacing
J.~Slater and N.~Masih, ``Cyclone amphan batters india and bangladesh, leaving
  at least 14 dead,'' \emph{The Washington Post}, 20 May, 2020. [Online].
  Available: \url{https://wapo.st/2TPKkC3}
\BIBentrySTDinterwordspacing

\end{thebibliography}

\appendix

%
%
%
%

\subsection{Proofs of Results on General One-Shot Feeding Problem}

%
%
%
  
%

\subsubsection{Proof of Proposition \ref{prop:necc-cond-opt-sol}}
\label{prf:necc-cond-opt-sol}

We first introduce the Lagrangian~$\Mc{L}$ for the
problem~\eqref{eq:equivoptmodel},
\begin{align}
  &\Mc{L}= J + \sum_{l \in V} \mu_l( \sum_{r,i} \fril-d_l ) +
    \sum_{r,i} \lambda_r^i(\sum_{l \in V_r^i} \fril - f_r) 
    +\sum_{l \in V} \gamma_l (\sum_{r \vert o_r = l}f_r - S_l) -
    \sum_{r}\delta_r f_r- \sum_{r,l,i} \delta_r^i(l) \fril  ,
\label{eq:lagrangian}
\end{align}
where $ \mu_l, \lambda_r^i, \gamma_l, \delta_r, \delta_r^i(l) \leq 0$
are the KKT multipliers. The stationarity and complementary slackness
conditions are
\begin{subequations}
  \begin{align}
    &\frac{\partial\Mc{L}}{\partial \fril} %
      = \bril + \mu_l + \lambda_r^i-\delta_r^i(l) = 0
      \label{eq:S1}
    \\
    &\frac{\partial\Mc{L}}{\partial f_r} %
      = -c_r -\sum_{i}\lambda_r^i+ \gamma_{o_r} -\delta_r = 0
      \label{eq:S2}
    \\
    &\mu_l( \sum_{r,i} \fril-d_l )= 0, \quad \lambda_r^i(\sum_{l \in
      V_r^i} \fril - f_r) = 0 \label{eq:C1}
    \\
    &\gamma_l (\sum_{r \vert o_r = l}f_r - S_l) = 0, \quad
      \delta_r^i(l) \fril = 0 , \quad \delta_r f_r = 0 . \label{eq:C2}
  \end{align}
\end{subequations}

\textbf{(a):} In an optimal solution, if $f_r^*>0$ then
$\delta_r^*= 0$. Further as $\gamma_{o_r}^*\leq 0$ and $c_r>0$, we can
use~\eqref{eq:S2} to obtain
\begin{align}\label{eq:lam-cond}
  \sum_{i= 1}^{\theta_r} \lstar = -c_r+\gamma_{o_r}^*< 0 .
\end{align}
Thus, we see from~\eqref{eq:C1} that
$\displaystyle \sum_{l \in V_r^i} \fril^* = f_r^* $ for at least one
leg $i$ in the route $r$. Hence we must have $\fril^* >0$ for some
$i \in \intrangecc{1}{\theta_r}$ and $l \in V_r^i$.  Now, if
$\fril^*>0$ then $\delta_r^i(l) = 0$. Hence condition~\eqref{eq:S1}
implies
\begin{equation}\label{eq:beta-cond}
  \bril=  -\lstar -\mu_l^*\geq -\lstar \geq 0, \text{ if } \fril^*>0
  .
\end{equation}

\textbf{(b):} We prove this by contradiction. Let us assume that there
exists some optimal solution such that~\eqref{eq:alloc-prop} does not
hold for a circular leg $i$ of some $r \in \mathbf{R}$. Then,
$\delta \ldef f_r^* - \sum \fril^*>0$ units of flow costs
$\delta c_r^i$ while earning no revenue.
	
Now consider a route $\bar{r}$ that follows the same sequence of nodes
as $r$ but without the $i^{\thh}$ leg of $r$. We can construct another
feasible solution with $f_r = f_r^*-\delta$ and
$f_{\bar{r}} = f_{\bar{r}}^* + \delta$ with all the other flows and
allocations remaining the same. Now leg $i$ of route $r$ in this
solution satisfies~\eqref{eq:alloc-prop}. Such a solution is feasible
and does not lose $\delta c_r^i$ while earning same revenue. Thus,
this solution earns more profit than the optimal solution, which is a
contradiction.

\textbf{(c):} Notice that $\fril^* \bril\geq - \fril^* \lstar$ for
each $\stup$, since if $\fril = 0$ then the inequality holds trivially
and if $\fril > 0$ then the condition~\eqref{eq:beta-cond} holds,
which implies
\begin{align*}
  \sum_{i = 1}^{\theta_r} \sum_{l\in V_r^i} \fril^* \bril %
  &\geq -\sum_{i = 1}^{\theta_r} \sum_{l\in V_r^i} \fril^* \lstar .
\end{align*}
Now, using part \textbf{(b)}, we obtain
\begin{align*}
  \sum_{i = 1}^{\theta_r} \sum_{l\in V_r^i} \fril^* \bril \geq -
  \sum_{l\in V_r^1} f_r^1(l)^* {\lambda_r^1}^* - f_r^* \sum_{i = 2}^{\theta_r}
  \lstar .
\end{align*}
From the K.K.T. conditions, we either have ${\lambda_r^1}^* = 0$ or
$\sum_{l\in V_r^1} f_r^1(l)^* = f_r^*$. In either case, we have
\begin{align*}
  \sum_{i = 1}^{\theta_r} \sum_{l\in V_r^i} \fril^* \bril \geq  - f_r^*
  \sum_{i = 1}^{\theta_r} \lstar  \geq f_r^* c_r ,
\end{align*}
where the second inequality follows from~\eqref{eq:lam-cond}.

\textbf{(d):} Now due to Part \textbf{(c)}, if $f_r^* > 0$ then
$\forall$ $i > 1$, $\exists$ $l \in V_r^i$ such that $\fril^*>0$. In
any feasible solution, for $i>1$, if $\bril < c_r^i$ and $\fril> 0$
then we can construct another feasible solution in which $\fril = 0$
keeping all other node allocation variables unchanged. The value of
the objective is \emph{strictly} more with such a new solution. Thus,
$\bril \geq c_r^i$ if $\fril^*>0$ for $i > 1$ in an optimal solution
of Problem \eqref{eq:equivoptmodel}. If $i=1$ and leg $1$ of route $r$
is circular then according to a similar argument as above, we must
again have at least one $l \in V_r^1$ such that $f_r^1(l)^* > 0$ and
for each such $l$, $\beta_r^1(l) \geq c_r^1$.

\textbf{(e):} In part~\textbf{(c)}, by setting $\theta_r = 1$, we get
\begin{align*}
  \sum_{l\in V_r^1} f_r^1(l)^* \beta_r^1(l) \geq f_r^*c_r .
\end{align*}
Similarly, setting $\theta_r = 1$ in~\eqref{eq:lam-cond} gives us
$\lambda_r^1 \leq -c_r < 0$, which along with~\eqref{eq:C1} means that
in any optimal solution, constraint~\eqref{eq:legconstr} is a strict
equality for single legged routes. Therefore, we have
$\sum_{l\in V_r^1} (\beta_r^1(l)-c_r)f_r^1(l)^* \geq 0 $. Now,
part~\textbf{(e)} follows in the same way as part~\textbf{(d)}. \qed

\subsection{Proofs of Results on Best Transportation %
  Parameters}

\subsubsection{Proof of Lemma~\ref{lem:viability}}
\label{prf:viability}

Consider an arbitrary node $l \in V$ and a service tuple $\stup$ such
that $l \in V_r^i$. Then, from~\eqref{eq:opt-pric}
and~\eqref{eq:Nodeprofitability} notice that
\begin{align*}
  \bril - c_r^i &= g_l(b) - \alpha\tril - c_r^i - k_l \leq g_l(b) -
                  g_l(1) - k_l
\end{align*}
where $g_l(b)$ is as defined in~\eqref{eq:alt-perceived-cost}. The
inequality follows from the fact that a leg is also single legged
route and therefore $g_l(1)$ lower bounds $\alpha\tril - c_r^i$. Thus
it is easy to see that $\bril - c_r^i \geq 0$ and hence
$\Mc{R}(\Mbf{R}) \neq \emptyset$ only if $g_l(b) \geq g_l(1) + k_l$.
Now, we show that if $\exists l \in V$ such that
$g_l(b) \geq g_l(1) + k_l$, then $\Mc{R}(\Mbf{R}) \neq \emptyset $. We
show that $r(l,1)^*\in\Mc{R}(\Mbf{R})$ for this case. Thus,
considering $q = r(l, 1)^*$, as per Table~\ref{table:1}, we have
\begin{align*}
  \beta_q^1(l) -c_q^1 = g_l(b) - \alpha t_q^1(l) -k_l -c_q^1 = g_l(b)
  - (\alpha t_q + c_q + k_l) = g_l(b) - (g_l(1) + k_l) \geq 0 .
\end{align*}
Thus, $q \in \Mbf{R}_1$ which implies
$\Mc{R}(\Mbf{R}) \neq \emptyset$, if $g_l(b)\geq g_l(1)+k_l$.

Now, we know from Remark~\ref{rem:alt-transport-b} that $g_l(b)$ is an
increasing function of $b$. Thus, for $k_l \geq 0 $, $b\geq 1$ is
necessary to ensure $\Mc{R}(\Mbf{R}) \neq \emptyset$. Again, for
$k_l<0$ for some $l \neq H$, and $b= 1$, $g_l(b)-g_l(1)-k_l \geq 0$,
following the discussion above.  \qed

\subsubsection{Proof of Proposition \ref{prop:multi-legged}}
\label{prf:multi-legged}

Consider a route $r \in \Mc{R}(\Mbf{R})$ with two distinct legs with
feasible services.  Note that such a route, has atleast one leg, $i$
with $o_r^i = D_r^i = H$. Also, route $r \in \mathbf{R}_2$
(see~\eqref{eq:R2}) as $\theta_r\geq 2$. Now, let $\bar{r}$ be the
route that is same as leg $i$ of route $r$. Thus,
$o_{\bar{r}} = D_{\bar{r}} = H$, $\theta_{\bar{r}} = 1$ and
$w_{\bar{r}}^1\geq 0$, which implies $\bar{r} \in \Mc{R}(\Mbf{R})$ as
$\bar{r} \in \mathbf{R}_1$ (see~\eqref{eq:R1}).
Hence, from~\eqref{eq:opt-pric} and~\eqref{eq:Nodeprofitability},
$\displaystyle w_{\bar{r}}^1 = \max_{l\in
  V_{\bar{r}}^1}\{\beta_{\bar{r}}^1(l)-c_{\bar{r}} \} \geq0$ implies
\begin{equation}
  g_l(b) \geq \alpha t_{\bar{r}}^1(l)+c_{\bar{r}} +k_l, \text{ for
    some } l \in V_{\bar{r}}, \ l \neq H, \label{eq:glb-lower-bnd}
\end{equation}
where we have used the fact that $g_H(b) = 0$ for all $b \geq 0$ and
$k_H > 0$, which imply that $l \neq H$. Now, note that for
\begin{align*}
  \alpha t_{\bar{r}}^1(l)+c_{\bar{r}}%
  &=\alpha t_{\bar{r}}^1(l) + c_{\bar{r}}(H,l) + c_{\bar{r}}(l,H)\\
  &\geq%
    \begin{cases}
      g_l(1) + c^*(H,l) , &\text{ if }\Mbf{R} = \Rin \\
      g_l(1) + c^*(l,H), &\text{ if }\Mbf{R} = \Rout
    \end{cases}\\
  &= g_l(1) + c^*(D,O)
\end{align*}
where we have split the route cost $c_{\bar{r}}$ into
$c_{\bar{r}}(H, l)$, the cost to go from $H$ to $l$ on route
$\bar{r}$, and $c_{\bar{r}}(l, H)$, the cost to go back from $l$ to
$H$.
Therefore, from~\eqref{eq:glb-lower-bnd},
$g_l(b)\geq g_l(1)+ c^*(D,O)+ k_l$, thus proving claim~(a).

Now using~\eqref{eq:alt-perceived-cost}, we get
$g_l(b) =\alpha t_{r(l,b)^*} + b c_{r(l,b)^*}$ and
$g_l(1) = \alpha t_{r(l,1)^*}+ c_{r(l,1)^*}$. Therefore, claim~(b) is
true.

Next, from Lemma~\ref{lem:viability}, we know that if $k_l > 0$ for
all $l \in V$ then $\Mc{R}(\Mbf{R})$ is non-empty only if $b >
1$. Then, notice from~\eqref{eq:best-alt-route}
and~\eqref{eq:alt-perceived-cost} that
$c_{r(l,0)^*}\geq c_{r(l,1)^*}\geq c_{r(l,b)^*}$ and
$t_{r(l,b)^*}\leq t_{r(l,\infty)^*}$. Using these facts we obtain
claim~(c). Finally, as $c_{r(l,b)^*}$ is a non-increasing function of
$b$ we conclude that $b_l^* \geq \hat{b}_l^*$.  \qed

\subsection{Proofs of Results on Feed-in Supply Optimization}

\subsubsection{Proof of Proposition \ref{prop:supply-opt}}
\label{prf:supply-opt}

Part \textbf{(a)} is true because Theorem~\ref{thm:algo-1} applies for
every possible supply configuration. We prove the remaining parts by
contradiction.

%

%
%
\textbf{(b)}: Let us assume there exists an optimal solution, which we
denote using a superscript $*$, in which $f_{r}^* > 0$ for a route
$r \in \MbrR$ such that $f_{r}^1(o_r)^* < f_{r}^*$. Then, there are
the following two cases. \textbf{(i)} $f_{r}^1(o_r)^* > 0$ and
$f_r^i(l)^* = 0$ for all other legs and nodes $(i,l)$ on route $r$;
and \textbf{(ii)} there is a pair
$(i_1, l_1) \in \intrangecc{1}{\theta_r} \times V_r$ with
$(i_1, l_1) \neq (1, o_r)$ and $f_r^{i_1}(l_1)^* > 0$. Note that
Proposition~\ref{lem:legs} implies that scenario (i) may occur only if
the route is simple ($\theta_r = 1$).

\textbf{(i)} In this case, clearly the original solution cannot be
optimal because $f_{r}^* - f_{r}^1(o_r)^*$ volume of vehicles simply
traverse the route without serving any demand, thus incurring a
non-zero cost while earning nothing.

\textbf{(ii)} Without loss of generality let $(i, l_1)$ be the leg,
node pair other than $(1, o_r)$ that first appears in the sequence
given by the route $r$ such that $f_r^{i}(l_1)^* > 0$. Then consider
the route $\bar{r}$, which is the sub-route of $r$ formed by excluding
all nodes in $r$ that occur prior to $(i, l_1)$ so that
$o_{\bar{r}} = l_1$. Thus, for each leg and node pair $(j,m)$ of the
route $\bar{r}$ there is a unique leg and node pair $(k,l)$ of route
$r$ such that $(j - i +1, m) = (k,l)$ and moreover they appear in the
same order. Now consider a solution $f_r = f_r^1(o_r) = f_r^1(o_r)^*$
and $f_{\bar{r}} = f_{\bar{r}}^* + (f_r^* - f_r^1(o_r)^*)$ and such
that $f_{\bar{r}}^j(m) + f_r^k(l) = f_{\bar{r}}^j(m)^* + f_r^k(l)^*$
for every $m \in V_{\bar{r}}^i$ and for each leg $j$ of route
$\bar{r}$. This solution is feasible and earns higher profits than the
original solution as the flow $(f_r^* - f_r^1(o_r)^*)$ does not have
to traverse the sequence of nodes from $(1, o_r)$ to $(i, l_1)$ and
the node allocations are unchanged. This again contradicts the
assumption that the original solution is optimal.  Thus, for each
$r \in \MbrR$, $f_r^1(o_r)^* = f_r^*$. As a consequence of this fact
and Proposition~\ref{lem:legs}, we also
have~\eqref{eq:no-redund-flow}. Also, $d_H= 0$ implies $f_r^* = 0$,
$\forall r $ s.t. $o_r = H$.

\textbf{(c)}: Assume an optimal solution with $f_r^*>0$, for a
$r\in \MbrR$ with $\beta_r^1(o_r)< c_r^1$. Using \textbf{(b)}, we know
$f_r^1(o_r)^* = f_r^*>0$. However, one could set $f_r^1(o_r)^* = 0$
and move the supply on $o_r$ as in part \textbf{(b)}. Then, the so
constructed solution would again earn higher profits, which
contradicts the assumption that the original solution is optimal.
Consequently, with Proposition \ref{itm:necc-c},
$\exists\ l \in V_r^i,\ s.t. \ \fril^*>0, \ \bril\geq c_r^i, \forall i
\in \intrangecc{1}{\theta_r}$ .

\textbf{(d)}:
%
Suppose that a node $l$ violates the property in~ \ref{A:4}
%
%
and yet $\fril^* > 0$ for some route $r$ and a leg $i$. Then by part
\textbf{(c)} we must have $\beta_r^i(l) - c_r^i\geq 0$. Now, given the
service tuple $\stup$, consider a simple route $q$, which is the
sub-route of route $r$ from the last visit to node $l$ in leg $i$ to
$H$ in that leg. As a result, $o_q = l$ and $q \in \Rin$. Now, notice
that
\begin{align*}
\beta_q^1(l) - c_q^1 \geq \beta_r^i(l) - c_r^i \geq 0 ,
\end{align*}
since $q$ is a sub-route of the leg $i$ of route $r$. This contradicts
the assumption that $l$ violates the property in~\ref{A:4}.

%
%

\textbf{(e)}: One can construct another route $q$ from $r$ by avoiding
the cycle. Again, moving the supply to this route (in a manner similar
to the previous parts) earns more profits.


\textbf{(f)}: In the scenario $s\leq \sum_l d_l$, the key observation
is that the full demand cannot be served by simple routes. Thus, if
$\displaystyle \sum_{r \vert o_r = l} f_r^* < S_l$ for some node
$l \in V$ then there is some unused supply. Such redundant supply
could potentially be used to serve more demand either at node $l$ or
moved to a different node $\bar{l}$ to meet the demand there with
simple routes. Assumption~\ref{A:4} implies that there exist simple
routes to which if the redundant supply is reallocated then the
profits are higher. This contradicts the assumption that the original
solution is optimal. Thus, in every optimal solution, we have
\begin{equation*}
S_l = \displaystyle \sum_{r \vert o_r = l} f_r^* \leq F_l^* \leq
d_l, \quad \forall l \in V ,
\end{equation*}
where the second inequality is just one of the constraints in the
optimization problem. \qed

\subsubsection{Proof of Lemma \ref{lem:Rl-least-perceived-cost}}
\label{prf:Rl-least-perceived-cost}

By the definition of $\Mc{S}(l)$ in~\eqref{eq:best-routes},
$\beta_r^1(l)-c_r > \beta_q^1(l)-c_{q}$, $\forall \ r \in \Mc{S}(l)$
and $\forall q \in ( \srs{l} \setminus \Mc{S}(l) )$. Every other route
$\bar{r} $ either originates from a different location or has multiple
legs and in each case the leg/route cost is higher and the operator
incomes are lower. Therefore,
$\exists q \in (\srs{l} \setminus \Mc{S}(l) )$ for which
$ \beta_q^1(l)-c_{q}> \beta_{\bar{r}}^i(l)-c_{\bar{r}}^i$. Hence,
$\beta_r^1(l)-c_r > \beta_{\bar{r}}^i(l)-c_{\bar{r}}^i$ for all
$r \in \Mc{S}(l)$ and $\bar{r} \notin \Mc{S}(l)$. Further, for
$\forall r \in \Mc{S}(l)$ and
$\forall q \in ( \srs{l} \setminus \Mc{S}(l) )$, the perceived costs
satisfy
\begin{align*}
  &( \alpha t_r + c_r ) - ( \alpha t_q + c_q ) = ( \alpha t_r + c_r -
    g_l + 1) - ( \alpha t_q + c_q  - g_l + 1)  =- (\beta_r^1(l) - c_r)
    + (\beta_q^1(l) - c_q) < 0 , \numberthis \label{eq:best-profits}
\end{align*}
where we have used~\eqref{eq:opt-pric}
and~\eqref{eq:Nodeprofitability} and the fact that $o_z = l$ for all
$z \in \srs{l}$, which implies that the route traversal time
$t_z = t_z^1(l)$. This proves the result. \qed


\subsubsection{Proof of Lemma~\ref{lem:equal_prices}}
\label{prf:equal_prices}

As for both the problems, the time window, best travel cost and the
best travel time are same, it therefore suffices to show that
$t_{\bar{r}}^{\theta_r-i+1}(l) = \htril$. Note that the $i^\thh $ leg
of $r$ is same as the $(\theta_r-i+1)^\thh$ leg of $\bar{r}$ in
reverse. Also, $T-t_{\bar{r}}^{\theta_r-i+1}(l)$ is the last possible
pick-up time along $(\bar{r}, \theta_r-i+1, l)$, based on the
observations of \eqref{eq:opt-pric} which implies that
$T-t_{\bar{r}}^{\theta_r-i+1}(l)$ is the first possible drop-off time
for the reverse route $r$ along the service tuple $(r,i,l)$. \qed

\end{document}